\documentclass[sigconf,edbt]{acmart-edbt2024}
\usepackage{amsthm}
\usepackage{amsmath}
\usepackage{color}
\usepackage{enumitem}
\usepackage{subcaption}
\usepackage{array}
\usepackage{booktabs}
\usepackage{makecell}
\usepackage{footnote}  
\usepackage{graphicx}
\usepackage{comment}
\usepackage{todonotes}
\usepackage{tikz}
\usepackage{multirow}
\usepackage{pgfplots}
\usepackage{environ}
\pgfplotsset{every tick label/.append style={font=\small}}
\usepackage{enumitem}
\usepackage{tabularx}
\usepackage{pbox}

\usepackage{tabularray}

\usepackage{algorithmicx}
\usepackage{algorithm}
\usepackage[noend]{algpseudocode}
\algrenewcommand\algorithmicrequire{\textbf{Input:}}
\algrenewcommand\algorithmicensure{\textbf{Output:}}

\newcommand\abol[1]{\textcolor{orange}{(Abol) #1}}

\newcommand\techrep[1]{#1}
\newcommand\submit[1]{}
\newcommand{\blue}{black}

\newcommand{\stitle}[1]{\vspace{2mm}\noindent{\bf #1}.}
\newcommand\eat[1]{}

\usepackage{xspace}

\newcommand{\dee}{\mathcal{D}}
\newcommand{\Gee}{\mathcal{G}}
\newcommand{\gee}{\mathbf{g}}

\newcommand{\el}{\mathcal{L}}

\newcommand{\g}{\mathbb{G}}
\DeclareMathOperator{\EX}{\mathbb{E}}

\newcommand{\at}[1]{{\tt \small #1}\xspace}

\newcommand{\bc}{{\sc Group-Coverage}\xspace}
\newcommand{\nbc}{{\sc Multiple-Coverage}\xspace}
\newcommand{\mc}{{\sc Intersectional-Coverage}\xspace}
\newcommand{\agg}{{\sc Aggregate}\xspace}
\newcommand{\css}{{\sc Classifier-Coverage}\xspace}

\newcolumntype{P}[1]{>{\centering\arraybackslash}p{#1}}

\usepackage{balance}

\def\BibTeX{{\rm B\kern-.05em{\sc i\kern-.025em b}\kern-.08em
    T\kern-.1667em\lower.7ex\hbox{E}\kern-.125emX}}
\usepackage{booktabs} 
\setcopyright{rightsretained}
\acmDOI{}
\acmISBN{XXXXXXXXXXXX}
\acmConference[EDBT 2024]{27th International Conference on Extending Database Technology (EDBT)}{25th March-28th March, 2024}{Paestum, Italy}
\acmYear{2024}
\eat{
\AtBeginDocument{%
  \providecommand\BibTeX{{%
    \normalfont B\kern-0.5em{\scshape i\kern-0.25em b}\kern-0.8em\TeX}}}
\setcopyright{acmcopyright}
\copyrightyear{2024}
\acmYear{2024}
\acmConference[EDBT'24]{The International Conference on Extending Database Technology}{March 25--28,
  2024}{Paestum, Italy}
\acmPrice{15.00}
\acmISBN{978-1-4503-XXXX-X/18/06}
}

\begin{document}

\title{Data Coverage for Detecting Representation Bias in Image Datasets: A Crowdsourcing Approach}
\thanks{This project was supported in part by NSF 2107290.}

\author{Melika Mousavi}
\email{mmousa7@uic.edu}
\affiliation{%
  \institution{University of Illinois Chicago}
  \country{}
}

\author{Nima Shahbazi}
\email{nshahb3@uic.edu}
\affiliation{%
  \institution{University of Illinois Chicago}
  \country{}
}

\author{Abolfazl Asudeh}
\email{asudeh@uic.edu}
\affiliation{%
\institution{University of Illinois Chicago}
\country{}
}

\begin{abstract}
  Existing machine learning models have proven to fail when it comes to their performance for minority groups, mainly due to biases in data. 
  In particular, datasets, especially social data, are often not representative of minorities.
  In this paper, we consider the problem of representation bias identification on image datasets without explicit attribute values. 
  Using the notion of data coverage\techrep{ for detecting a lack of representation}, we develop multiple crowdsourcing approaches.
  Our core approach, at a high level, is a divide and conquer algorithm that applies a search space pruning strategy to efficiently identify if a dataset misses proper coverage for a given group. We provide a different theoretical analysis of our algorithm, including a tight upper bound on its performance which guarantees its near-optimality.
  Using this algorithm as the core, we propose multiple heuristics to reduce the coverage detection cost across different cases with multiple intersectional/non-intersectional groups.
We demonstrate how the pre-trained predictors are not reliable and hence not sufficient for detecting representation bias in the data. Finally, we adjust our core algorithm to utilize existing models for predicting image group(s) to minimize the coverage identification cost.
We conduct extensive experiments, including live experiments on Amazon Mechanical Turk to validate our problem and evaluate our algorithms' performance.
\end{abstract}




\maketitle

\vspace{-2mm}\noindent
\textcolor{\blue}
{\textbf{Artifact Availability:}\\
The source code, data, and/or other artifacts have been made available at \url{https://github.com/melimou/ImageDataCvgCrwd}.
}
\vspace{-2mm}\section{Introduction}\label{sec:intro}
Tracing back machine bias to its source, there have been major efforts to identify different types~\cite{mehrabi2021survey, olteanu2019social,friedman1996bias} and sources~\cite{torralba2011unbiased,crawford2013hidden,diakopoulos2015algorithmic} of bias in data.
{\em Representation bias}~\cite{shahbazi2023representation}, in particular, happens when a dataset fails to represent some parts of the target population \cite{suresh2021framework}. 
Lack of representation from certain minority groups in data has caused many instances of machine bias and algorithmic unfairness in data-driven algorithms.
For example,\techrep{ on multiple occasions, Google's image search results have been reflecting societal biases, with the most famous example being the ``CEO'' search query returning only pictures of male CEOs in the top results \cite{kay2015unequal}. Other notable examples include }
Facebook's ad algorithm excluding women from seeing specific jobs~\cite{imana2021auditing}, or commercial gender classification systems from Microsoft, IBM, and Face++ that performed up to 35\% worse on dark skin women compared to light skin men \cite{buolamwini2018gender}.
\techrep{
Another example is the infamous ``Google gorilla'' incident \cite{google-gorilla} where an early image recognition algorithm released by Google had not been trained on enough dark-skinned faces and failed to label black females appropriately.
Last but not least, 
}\submit{Similarly, }
the blink-detection feature of Nikon Cameras misclassified Asian eyes as being closed~\cite{closed-eyes} due to a lack of representation for this group.
\textcolor{\blue}{
We shall also demonstrate similar results in our experiments in \S~\ref{exp:validation:downstream}. 
}

Recognizing the potential harms of representation bias, data coverage~\cite{asudeh2019assessing,lin2020identifying,asudeh2021coverage,tae2021slice,accinelli2021impact,10.14778/3415478.3415486,accinelli2020coverage,jin2020mithracoverage} has been introduced to ensure proper representation of minority groups in datasets used for decision making and building advanced data science tools.
At a high level, a dataset has proper coverage for a given group if it contains at least a certain amount of samples belonging to that group.
\techrep{
With many angles to tackle, the problem of identifying and resolving insufficient coverage has been studied for datasets with discrete~\cite{asudeh2019assessing} and continuous~\cite{asudeh2021coverage} attributes populated in single or multiple \cite{lin2020identifying} relations. 
}
Despite the extensive work 
to detect lack of coverage in a given dataset, {\em existing work is limited to the structured data}, in the form of a table where every row is a tuple of numeric attributes. On the contrary, many of the well-known incidents of representation bias causing machine bias, including all aforementioned examples, are in non-tabular contexts, such as multi-media or textual data. 

Admitting the wide range of multimedia databases, with attributes of interest being in different forms and cardinalities, as our first attempt in this project, we consider {\em image data and a small number of low-cardinality categorical attributes}. Our choices are motivated by image data’s popularity in data science tasks and the reported unfairness issues in the image application domains.
Our assumption of the attributes of interest follows the fact that sensitive attributes such as {\tt\small race} and {\tt\small gender} are low-cardinality and non-ordinal, where each value such as {\tt\small black} or {\tt\small female} represents a specific demographic group.

It is common that image datasets usually \underline{lack explicit values} for attributes of interest (such as {\tt\small gender} or {\tt\small race})\techrep{, crucial for coverage identification}. An image dataset is often a collection of images from different domains with little to no information about their domain and which groups they belong to. As a result, even studying coverage over low-cardinality and categorical attributes of interests is challenging in these cases. 

There are multiple directions one can seek to overcome such challenges. Considering a small number of categorical attributes of interest (such as {\tt \small race} and {\tt \small gender}), one can use off-the-shelf automated techniques, such as classifiers, to first label tuples with their demographic information\footnote{In presence of accurate predictive models, {\em our algorithms utilize them} to identify coverage with minimum cost to verify the correctness of their results (see \S~\ref{sec:classifier}).}. Then, relying on the predicted groups, apply the coverage detection techniques to identify the lack of coverage in data.
However, as we observe in our experiments, this approach fails, mainly due to the following issues:

\begin{enumerate}[leftmargin=*]
    \item \textit{(Machine Bias)}: while the objective of identifying lack of coverage is to minimize machine bias, using (problematic) off-the-shelf models will transfer their biases into the labeled data, causing bias in the evaluation of the dataset. For example, consider a gender-detection classifier. Due to the inherent issues in how the classifier has been trained (and the data it used), it may perform differently across different minority groups. For instance, in our experiments (Table~\ref{table:classifier}), we observed that the precision of a gender classifier from a well-known face recognition framework such as DeepFace \cite{serengil2020lightface} for females can get as low as 8\% for a given image dataset. 
    \techrep{
    These observations along with many real-life examples of classifiers' failure to perform well for minority groups further support the idea that there is no guarantee that relying on the existing models for the purpose of coverage identification leads to precise and robust outcomes.
    }
    \item \textit{(Lack of distribution generalizability)}:  Existing tools are trained using data that may come from a different application domain, following a different distribution, and hence may not perform well on the dataset to be evaluated. Let us consider the example of a gender-detection classifier once again. Suppose the classifier has been trained using the standard portray images with a solid background. One cannot expect the classifier to perform well on randomly taken images~\cite{kulynych2022you}. Note that applying transfer learning techniques to retrain the model using the dataset to be evaluated is not helpful since we cannot expect a model to identify the representation biases of the dataset it is trained on.
\end{enumerate}
Considering the above issues with the existing data-driven tools and techniques, a promising approach to consider is crowdsourcing: to efficiently use human workers to identify a lack of coverage issues. Crowdsourcing is particularly promising for image data, for the tasks such as image labeling, which while being challenging for the machine, are considered "easy" for human-being to conduct with minimal error\footnote{\textcolor{\blue}{We do not make the assumption that human beings are completely reliable but rather suggest our framework for cases where human labels are more reliable.}}. \textcolor{\blue}{Using crowdsourcing for labeling the images with their attributes of interest can potentially add human bias into the process. Fortunately, accurate and reliable crowdsourcing that minimizes individual errors and biases has been studied well in the literature.} Aggregating the responses of multiple crowd workers \cite{ipeirotis2010quality,dawid1979maximum,10.1145/2806416.2806451}, and profiling the crowd \cite{whitehill2009whose,welinder2010multidimensional,wang2011managing} are some of the known techniques proposed for this purpose.
A baseline solution then can be designed as a two-step process: first ask the crowd to provide the attribute values for all images in the dataset. Then apply off-the-shelf coverage identification techniques \cite{asudeh2019assessing} to detect the uncovered groups.
Cost-effectiveness, however, is a major requirement in crowdsourcing frameworks such as Amazon Mechanical Turk since there usually is a cost associated with each crowd task. The proposed baseline solution is ineffective in such frameworks because, depending on the size of datasets, it may require a significant number of tasks, meaning a considerable cost to study coverage in a given dataset.
Hence, in this paper we study the problem of {\em identifying the lack of coverage in an image dataset with the minimum number of crowd tasks}.
\textcolor{\blue}{
Existing research on bias detection in image data sets is limited to \cite{hu2020crowdsourcing} that proposes a crowd-sourcing workflow to facilitate discovering attributes with potential sampling bias (e.g., finding out the airplanes in an image data set are facing right).
The high level idea in this paper is to show random samples of images from an input dataset to the crowd to identify comon similarities, and then ask the crowd judge to verify the discovered statements (see \S~\ref{sec:related} for more details).
In contrast, our objective in this paper is to detect representation bias in terms of the data coverage with respect to the given attributes of interest such as \at{gender}, instead of discovering attributes that may reflect potential sampling bias.
}

\vspace{-2mm}
\stitle{Summary of Contributions} We consider the problem of coverage identification in image datasets using crowdsourcing. 
\textcolor{\blue}{To the best of our knowledge, our paper is {\em the first to study data coverage in image datasets}.}
In summary, our contributions are the following:



\begin{itemize}[leftmargin=*]
    \item We propose a divide-and-conquer (d\&c) algorithm to identify the coverage of a demographic group across an image dataset. 
    \textcolor{\blue}{
    To enable the development of our algorithm, we employ set-based crowd tasks, which have been utilized in various crowdsourcing studies~\cite{li2017crowdsourced,li2016crowdsourced,marcus2012counting,gomes2011crowdclustering}.
    At a high level, our algorithm falls in the general class of group testing approaches, wherein the process of identifying certain objects is divided into tests on groups of items~\cite{du2000combinatorial,dorfman1943detection}.
    While, following the same logic of the group testing approaches, the design details of our algorithm are problem specific, making its performance close to the lower bound on the maximum number of tasks. We prove this by 
    } a tight upper bound on the maximum number of tasks the algorithm generates. 
    \item Using our d\&c algorithm as the core, we propose efficient algorithms for coverage identification over different scenarios with multiple non-intersectional and intersectional groups.
    We further introduce practical heuristics for coverage identification by carefully aggregating the minority groups into the so-called ``super-groups''.
    \item In presence of pre-trained classification models that predict the group(s) an image belongs to, we adjust our core algorithm to utilize their prediction and minimize the coverage identification cost. In cases where the models accurately predict the group labels, our algorithm only generates a small number of tasks to verify the correctness of the results.
    \item We evaluate our algorithms using extensive experiments on real and synthetic settings.
    We run {\em live experiments on Amazon Mechanical Turk with real workers} to validate our proposal. Besides, our performance evaluation experiments verify our theoretical findings, confirming the effectiveness of our algorithms. 
\end{itemize}


\vspace{-3mm}
\section{Preliminaries}

\subsection{Data Model}
We consider a dataset in form of a collection $\dee$ of $N$ objects, each being an image. 
For example, $\dee$ can be a set of $N=10,000$ human face images.
We use the notation $t_i$ to refer to the $i$-th object in $\dee$.
We consider the {\em no explicit attribute-value} model for the data. 
That is, $\dee$ only contains the objects, while the objects are not annotated.

We assume objects are associated with at least one attribute of interest considered for identifying representation bias.
Formally, we use $\mathbf{x}=\{x_1,\cdots,x_d\}$ to specify the attributes of interest.
Each attribute of interest is a categorical sensitive attribute such as {\small \tt race}, {\small \tt gender}, and {\small \tt age-group}. Each attribute has a cardinality of two or more, specifying different non-overlapping (demographic) groups.
For example, a binary attribute {\small \tt gender} with values {\small \tt \{male, female\}} partitions the individuals into two non-overlapping groups.

In particular, we consider three different scenarios with attribute and group models:
The most simple scenario is the {\em single binary} sensitive attribute case where objects are associated with only one binary sensitive attribute. Many of the problematic representation bias cases that have been reported fall under this category.
Examples of this type of attribute include {\small \tt skin-tone} aggregated into a binary feature of {\small \tt fair} and {\small \tt dark} skin-tone \cite{celis2020implicit}. Studies showed that the pulse oximeter devices have a questionable accuracy in measuring arterial oxyhemoglobin saturation in individuals with dark skin-tone \cite{feiner2007dark}, which proves it is imperative that skin-tone feature be taken into account in developing this type of device.

The immediate generalization is the {\em multiple non-intersectional} groups case where each object is associated with one sensitive attribute with cardinality larger than two.
A non-binary attribute {\small \tt race} with values {\small \tt \{White, Black, Hispanic, Asian, Others\}}, or a multi-valued attribute {\small \tt gender} with values {\small \tt \{male, female, non-binary\}} are examples of this case.

The next level of generalization is the intersection of multiple attributes where each object is associated with more than one sensitive attribute that can be either binary or non-binary. An example of this case can be the intersection of {\small \tt race} and {\small \tt gender}, where each individual can be associated with one value from each of these attributes such as {\small \tt Asian female} or {\small \tt Hispanic male}.

 \vspace{-2.8mm}
\subsection{Data Coverage}

We use the notion of {\em data coverage} \cite{asudeh2019assessing} to identify representation bias in a dataset $\dee$.
In particular, consider a dataset $\dee$ with $d$ attributes of interest $\mathbf{x}$, a count threshold $\tau$ (e.g. $\tau=50$),
and a subgroup $\gee$ (e.g. {\small \tt \{gender=male, race=white\}}) defined over $\mathbf{x}$.
The dataset satisfies coverage over $\gee$, if there are at least $\tau$ objects in $\dee$, matching the subgroup $\gee$ (e.g. there are more than 50 objects with {\small \tt gender=male} AND {\small \tt race=white}).

For datasets with more than one attribute of interest ($d>1$), \textit{patterns} are used to specify the subgroups.
A pattern $P$ is a string of $d$ values, where $P[i]$ is either a value from the domain of $x_i$, or it is ``unspecified'', specified with $X$. 
For example, consider a dataset with three binary attributes of interest $\mathbf{x}=\{x_1, x_2, x_3\}$. The pattern $P=X01$ specifies all the tuples for which $x_2=0$ and $x_3=1$ ($x_1$ can have any value).
Consider the universe of all patterns over a set of attributes $\mathbf{x}$.
We say a pattern $P$ is a parent of another pattern $P'$ if (a) there exists exactly one attribute $x_i$ on which $P$ and $P'$ are different, while (b) $P[i]=X$ (unspecified).
Note that $P$ in this case is a more general subgroup than $P'$, since all the objects matching $P'$ also match $P$, but the vice-versa is not valid.
A pattern $P$ is a {\em maximal uncovered pattern} (MUP) if there are less than $\tau$ objects in $\dee$ matching it, while all of its parents are covered.
The lack of coverage in a dataset is identified by discovering all of its MUPs.

\vspace{-2.8mm}
\subsection{Crowdsourcing Model}\label{sec:pre-crowd}
A major challenge in studying coverage over multimedia data is that the objects are not annotated and we do not know the values on $\mathbf{x}$ beforehand.
We use crowdsourcing to overcome this challenge. 
In crowdsourcing platforms such as Amazon Mechanical Turk (AMT) or Crowdflower, a {\em microtask} or a {\em Human Intelligence Task (HIT)} is a simple task that usually requires no domain-specific knowledge or expertise, has a clear description and a price which workers can accept and complete and get paid given their result is approved by the requester of the task.

\stitle{Quality control and aggregation model}
Quality of answers is a well-studied crowdsourcing challenge.
One popular approach is to employ a redundancy-based strategy in which a single HIT is assigned to multiple workers and the correct answer (the {\em truth}) is inferred by aggregating the multiple answers. There are several studies on truth inference methods. 
The proposed techniques in this paper are {\em agnostic} to the
choice of the crowdsourcing framework, quality control, and HIT aggregation model. 
In our experiments, we adopt the popular majority vote strategy~\cite{10.14778/3055540.3055547} to get the truth and Qualification and Rating \cite{daniel2018quality} as individual assessments to ensure a higher quality of answers from the crowd.

\begin{figure}[!tb]
    \centering
    \includegraphics[width=.49\textwidth]{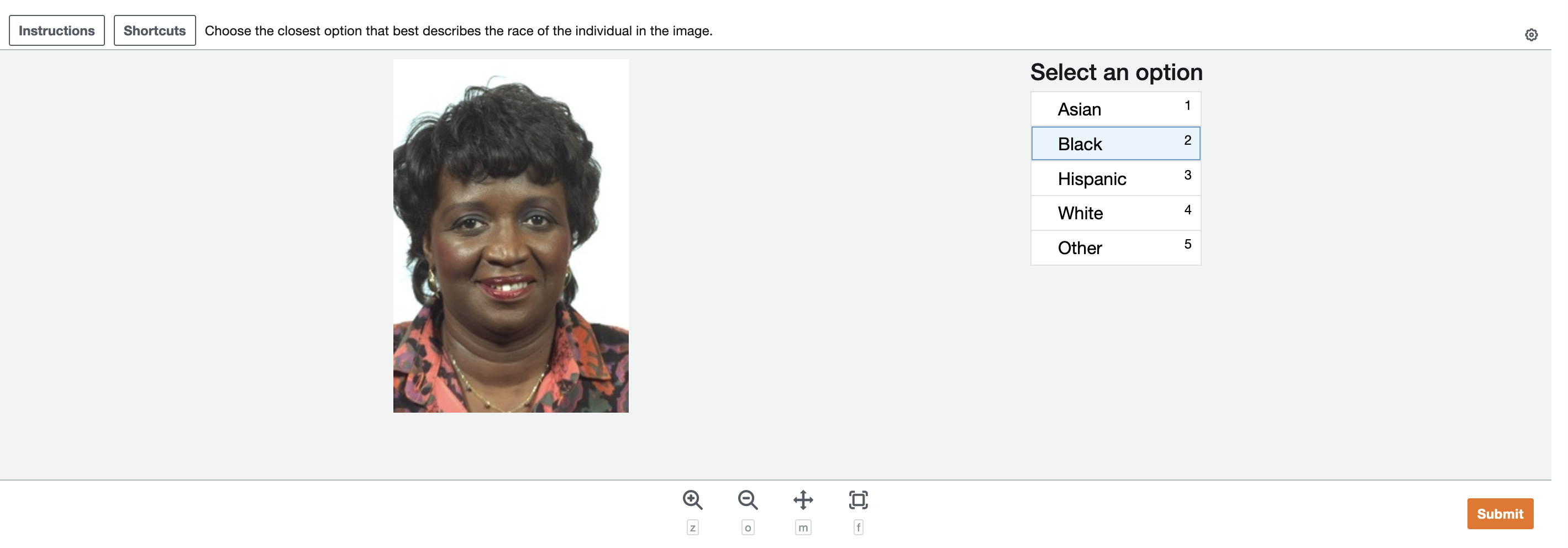}
    \vspace{-8mm}\caption{An example of point query to label a {\small \tt race} attribute}
    \label{fig:pquery}
    \vspace{-4mm}
\end{figure}

\stitle{Query/ HIT model}
We consider two types of queries:
\begin{enumerate}[leftmargin=*]
    \item {\em Point queries:} A point query is a request to provide a piece of information (attribute value) about a single object. The query itself can be either a yes-no question or providing one or more labels associated with the attributes of interest. In Figure \ref{fig:pquery} an example of a point query is demonstrated, where the worker is asked to provide the race of an individual.
    \item {\em Set queries:} 
    While a point query asks the crowd to provide specific attribute values for specific objects, the set queries are for the purpose of {\em verification}.
    A set query is a simple yes-no question about a given set of multimedia objects (image, video, etc.). 
    That is to ask if the set contains at least one object belonging to a specific (sub-)group. For example, Figure~\ref{fig:squery} shows a set query asking the crowd to verify if the set contains any females.
    In practice, one may need to consider an upper bound on the number of objects in a set query 
    for the query to be reasonable and the answers more accurate. 
\end{enumerate}


\stitle{Pricing model}
Pricing and incentive models for crowdsourcing frameworks have been extensively studied. Examples of such methods are {\em fixed price models}, {\em bidding models} \cite{10.1145/2488388.2488489}, and {\em posted price models} \cite{10.1145/2488388.2488490}. 
In this paper, we adopt the fixed pricing model, that is all tasks have an equal cost.
Therefore, our objective is to minimize the number of tasks required to finish the task of detecting coverage which in turn is in line with minimizing the total cost.
    \begin{figure}[!tb]
    \centering
    \techrep{\includegraphics[width=0.4\textwidth]{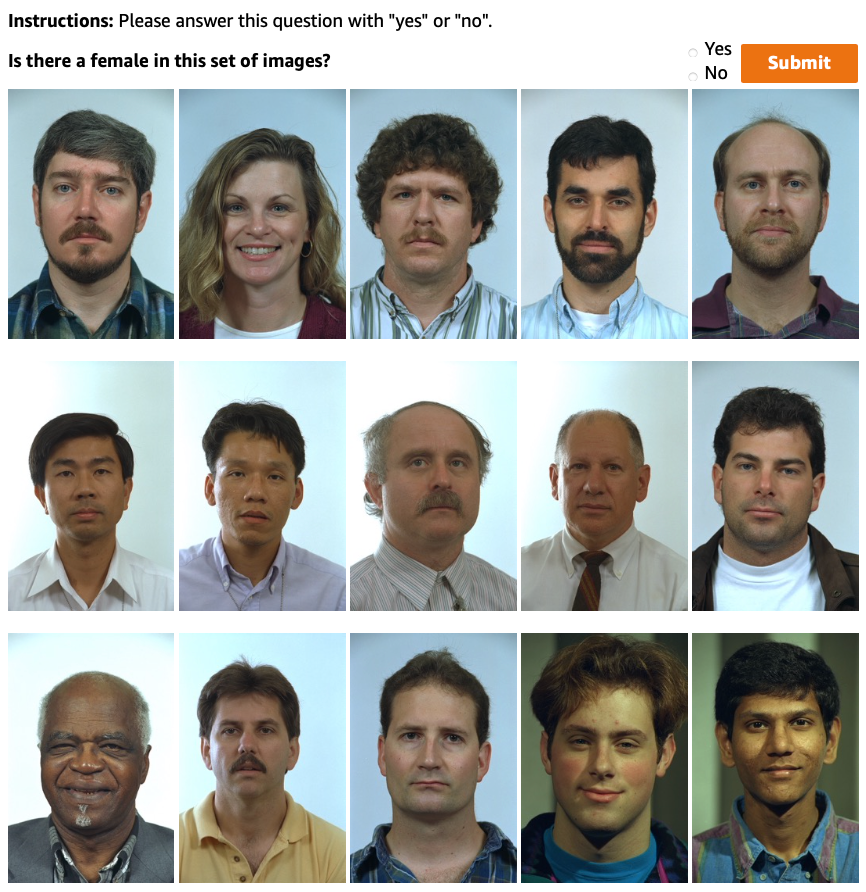}}
    \submit{\includegraphics[width=0.4\textwidth]{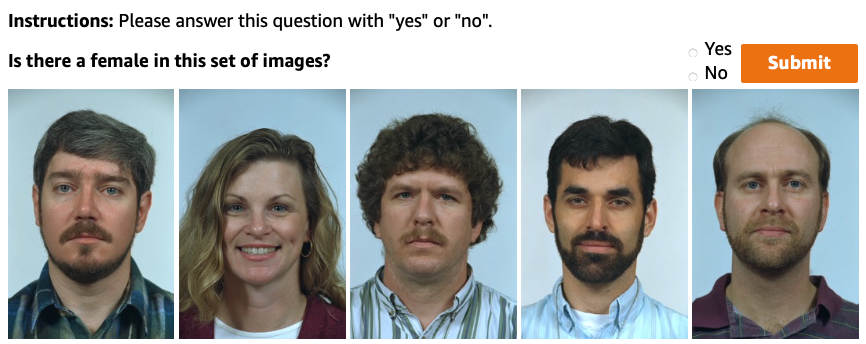}}
    \vspace{-3mm}\caption{An example of set query about {\small \tt gender} attribute}
    \label{fig:squery}
    \vspace{-6mm}
    \end{figure}

\vspace{-2.5mm}
\subsection{Problem Definition} \label{s:problem_def}
Having explained the data model, data coverage, and the crowdsourcing model, we now define our problem as follows: 

\begin{center}
\framebox[\columnwidth]{\parbox{0.9\columnwidth}    { 
        \textbf{\small \textsc{Problem Formulation:}
        }
        \textit{\small 
            Given an image dataset $\dee$,
            the attributes of interest $\mathbf{x}$,
            and the coverage count threshold $\tau$,
            identify the lack of coverage on $\dee$ while minimizing the number of tasks required.
        }
    }}
\end{center}

In particular, we study this problem under three settings: (1) single (minority) group, (2) multiple non-intersectional groups, and (3) intersectional groups, proposing efficient solutions tailored for each setting.
At a high level, our main idea is to divide the dataset into subsets of a specified size, ask the crowd a question about the attributes of interest, and based on the crowd's answer to the question, decide to divide that particular subset into two halves or prune it. Modeling each algorithm's flow into a binary tree, 
we analyze the efficiency of our proposal particularly when the dataset size is very large.
We also study the problem in presence of predictive tools for labeling the data, utilizing them to effectively form the tasks, and minimizing the interaction with the crowd.

\vspace{-2.5mm}
\section{Efficient Coverage Identification}\label{s:binary}

We start our technical sections by designing a general algorithm that can be used for detecting coverage over different settings of sensitive attributes.
In particular, {\em given a demographic (sub)group $\gee$, our goal in this section is to identify if $\gee$ is uncovered}.

Before proposing the algorithm, however, let us consider the single binary sensitive attribute case to observe a challenge we need to address in our algorithm.
In such a setting, 
the binary attribute partitions the data into two groups: the majority group (to which most of the objects belong) and the minority group. Let $\gee_1$ and $\gee_2$ be the majority (e.g. {\small \tt gender=male}) and
minority (e.g. {\small \tt gender=female}) group, respectively.

\stitle{Challenge} We observe that {\em verifying that $\gee_1$ is covered} can be done efficiently.
To see why, suppose the coverage threshold is $\tau = 100$, i.e., a group is covered if there are at least $100$ instances of it in the dataset.
Assume the (majority) group $\gee_1$ contains $n_1 \gg 100$ objects in the dataset. 
In order to verify that $\gee_1$ is not uncovered, it is enough for the crowd to discover 100 of those objects, not the entire $n_1$. 
Following this, $\Theta(\tau)$ provides a lower bound on the number of crowd tasks required for verifying the coverage for a given group. 
Still, this lower bound only holds for the groups that are covered, i.e. there at least $\tau$ of those in the dataset.
Surprisingly,
unlike the majority group,
verifying that a minority group is indeed uncovered is cumbersome.
This is because even though discovering $\tau$ objects from a group is enough for verifying that it is covered, one cannot verify a group is uncovered until there is a chance that the dataset might still have enough objects from that group. Thus, assuming a non-zero probability for each unlabeled object to belong to each group, {\em one might need to ask the crowd to label the entire dataset before one can confirm that a specific group is uncovered}. 

\vspace{-2.5mm}
\subsection{Coverage Identification for a Given Group}
Verifying that a minority group is uncovered is challenging. Our idea is to
design {\em a divide and conquer algorithm} that, instead of {point queries}, uses {\em set queries} to iteratively eliminate subsets of data that {\em does not include any object from the given group}.
At a high level, the algorithm asks a set query from the crowd, inquiring whether the selected set contains at least one object from the given group $\gee$ (Figure~\ref{fig:squery}).
The user may provide two responses (yes/no). 
Interestingly, {\em in either case}, the user response provides useful information that helps in the efficient study of the coverage:
\begin{itemize}[leftmargin=*]
    \item {\em The answer is ``No'':}
    If the answer to a set query is no, it means the set does not include any object from the given group $\gee$. As a result, the algorithm can safely prune the set, asking no further questions about it.
    In particular, for a group that is not covered, one can expect to see no answers on large set queries helping to quickly prune a large portion of the dataset.
    \item {\em The answer is ``Yes'':} A yes answer to a set query means that the set contains one or more objects from the group $\gee$. Therefore, the algorithm cannot prune the subset since it can have any number (larger than zero) of the objects in $\gee$.
    At the first glance, the queries with yes answers do not provide useful information as the algorithm cannot prune the subset (it needs to divide it to smaller subsets).
    However, a key observation is that, \textcolor{\blue}{since the sets are disjoint,} {\em the algorithm will require observing only a limited number of yes answers} before it stops.
    That is because the number of such queries provides a {\em lower-bound} on the number of objects from $\gee$ in the dataset. Hence, as soon as the lower-bound reaches to $\tau$, the algorithm can stop, knowing that $\gee$ is covered.
\end{itemize}

Based on the observations on the answers of the set query tasks, we design our divide and conquer algorithm as follows. 

We use a {\em binary tree} data structure to efficiently implement the algorithm.
Each node in the binary tree has the following structure:
\begin{center}
    \small
    \begin{verbatim}
    struct node:
        b_index     // the beginning index of the range
        e_index     // the end index of the range
        parent=null,      // link to the parent node
        left=null,        // link to the left child
        right=null,       // link to the right child
        checked=false, // true if at least one of its
                    // child nodes has returned a yes answer
    \end{verbatim}
    \vspace{-3mm}
\end{center}
Each node in the tree is associated with a set query containing the objects $\{t_{b\_index},\cdots ,t_{e\_index}\}\in\dee$.
In addition, every node in the tree has pointers to its parent and children.
Finally, every node contains a boolean variable \textcolor{\blue}{\tt\small ``checked''} (with the default value {\tt\small false}) that is used for maintaining a lower-bound on the number of objects discovered from the target group $\gee$.

\begin{algorithm}[!tb]
    \caption{\bc}\label{alg1:coverageBinary}
    \begin{algorithmic}[1] \small
    \Require{Dataset $\dee$, dataset size $N$, subset size upper bound $n$, coverage threshold $\tau$, target group $\gee$}
    \Ensure{Coverage of group $\gee$, \textcolor{\blue}{the count lower-bound $cnt$}}
        \State $cnt \gets 0$;
        Let $Q$ = an empty queue
        \For{$i\gets0$ to $N$ with step size $n$}:   {\tt\small // init roots of subtrees}
            \State root $\gets$ node$(i,i + n)$;
            $Q.add(root)$
        \EndFor
        \While{$Q$ is not empty}
            \State $T\gets Q.del\_top()$
            \State $(i,j)\gets (T.b\_index,T.e\_index)$
            \State ans $\gets ${\sc AskQuestion}$(\{t_i,\cdots,t_j\},\gee)$ {\tt\small // set query}
            \If{$T.$parent is null}
                \State {\bf if} ans=yes {\bf then} $cnt \gets cnt+1$
                {\bf else continue}
            \Else
                \If{ans=no} \textcolor{\blue}{\tt\small // prune} \label{line11}
                    \State {\bf if} {\it T = T}.parent.left {\bf then} $T \gets Q.del(T.$parent.right$)$ \textcolor{\blue}{\tt\small // observe from Line~\ref{line:last} that left nodes are added first to $Q$}\label{line1}
                    \State {\bf else continue}
                \EndIf
                \State {\bf if} $T.$parent.checked {\bf then} $cnt \gets cnt + 1$ \label{line2}
                    \State {\bf else} $T.$parent.checked$\gets$ {\bf true}
            \EndIf
            \State {\bf if } $cnt=\tau$ {\bf then return true} {\tt\small // coverage threshold satisfied}
            \If{$j>i$} \textcolor{\blue}{\tt\small //if setsize>1}
                \State $T.$left$\gets$ node$(i,\lfloor\frac{i+j}{2}\rfloor)$;
                $T.$right$\gets$ node$(\lfloor\frac{i+j}{2}\rfloor+1, j)$
                \State $T.$left.parent$\gets T.$right.parent$\gets T$
                \State $Q.add(T.$left$)$; $Q.add(T.$right$)$ \label{line:last}
            \EndIf
    \EndWhile
    \State {\bf return false}, $cnt$ {\tt\small //uncovered}
    \end{algorithmic}
\end{algorithm}

Using the tree data structure,
Algorithm~\ref{alg1:coverageBinary} shows the pseudo-code of our proposed algorithm for the single binary sensitive attribute case.
The algorithm uses the variable {\tt \small cnt} to maintain the lower-bound number of objects discovered from the target group $\gee$.
Considering a maximum size of $n$ for the set queries, the algorithm starts by partitioning the data into $\lfloor N/n\rfloor$ subsets, allocating each to a binary tree.
Adding the roots of the trees to the queues, the algorithm then iteratively removes a tree node from the queue until a lower-bound count of $\tau$ is achieved for the minority group $\gee$ or it verifies that $\gee$ is uncovered.

For every node $T$ removed from the queue, the algorithm asks the crowd to verify if its corresponding set contains at least one instance belonging to $\gee$.
Depending on the answer from the crowd, multiple situations can happen. If the
node stands for the root of a binary tree and the response is yes, the algorithm has found at least one more object from $\gee$; but if the response is no (\textcolor{\blue}{Line~\ref{line11}}), it can safely prune the entire set from the search space and continue with other sets.
On the other hand, if $T$ is not a root node, if the answer is no, it means the answer for the other child of its \textcolor{\blue}{parent} should be yes.
That is because the parent node contained at least one object from $\gee$. Since $T$ does not contain any such object, the other child of its parent should contain at least one. As a result, in \textcolor{\blue}{Line~\ref{line1}}, the algorithm replaces $T$ with the other child of its \textcolor{\blue}{parent}, safely knowing the answer to a query to the new $T$ is yes.

When the answer to a non-root node query is yes, the algorithm may or may not be able to increase the lower-bound $cnt$. Note that the algorithm has already associated at least one object from $\gee$ to the parent of each non-root node (the nodes with no answers have been pruned). As a result, the lower bound on $\gee$ gets increased if the answer to both children of a parent node is yes. We use the variable {\tt\small checked} for this purpose. {\tt\small checked} is true, if the answer to one of the children of a node is yes. Using this, when receiving a yes response, the algorithm increases the lower bound (\textcolor{\blue}{Line~\ref{line2}}) only if the {\tt\small checked} variable of the parent of $T$ is true.
At any moment that the lower-bound reached the threshold $\tau$, the algorithm stops marking $\gee$ as covered.
Finally, 
the algorithm breaks the yes nodes with set sizes larger than one in two halves, adding them to the queue.
If after checking all nodes in the queue the threshold $\tau$ is not reached, $\gee$ is uncovered.

\color{\blue}
\stitle{Running Example} To better demonstrate \bc, let us consider a toy example with 16 images, where each image belongs to either group $\Box$ or group $\bigtriangleup$. Suppose we would like to check if $\bigtriangleup$ is covered, while $\tau=3$.
Figure~\ref{fig:runningexample} shows the The binary tree representation of the \bc algorithms' flow, while the root of the tree shows the entire image set.
The answer to first query on the root is \at{yes} so the lower-bound value $cnt$ gets updated to 1 and the images gets divided in two halves. The answer to both children of the root are also \at{yes}, so $cnt$ gets updated to 2, while each set gets divided by half. The next level of the tree contains four set queries, each containing four images. Moving from left to right, the answer to the left-most query is \at{no}, therefore, (i) this set gets pruned and (ii) the algorithm {\em without issuing a new task} knows the answer to the second-from-left query is \at{yes} (otherwise its parent query could not be \at{yes}). The same situation happens for the two right nodes.
Next, in the fourth level, the algorithm issues the first two queries from left and since the answer to both is \at{yes}, $cnt$ gets updated to 3 and the algorithm stops since it reaches the coverage threshold $\tau$. Note that in this example, the algorithm issues seven queries to the crowd before it stops.

\color{black}

\vspace{-3mm}
\subsection{Algorithm Analysis} \label{single_analysis}

\begin{lemma}\label{lemma1}
(Correctness) The \bc algorithm successfully identifies if a group $\gee$ is covered or not, i.e., if there are at least $\tau$ instances of $\gee$ in the dataset $\dee$.\submit{\footnote{Due to the space limitations, the proofs are provided in the technical report~\cite{techrep}.}}
\end{lemma}

\techrep{
\begin{proof}
\vspace{-3mm}
Each set query with a {\tt\small yes} answer contains at least one object belonging to the group $\gee$.
Using this, the algorithm maintains a lower bound $cnt$ on the $|\gee|$ in $\dee$. That is, $cnt\leq |\gee|$.
The algorithm returns {\tt\small true} when $cnt=\tau$.
When $cnt=\tau$, $\gee$ is covered, because $\tau = cnt \leq |\gee|$.
The algorithm returns false when the queue is empty and $cnt<\tau$.
For sets with {\tt\small yes} answers, the algorithm divides the set in two halves, unless the set size is $1$. As a result, when the queue, all the set questions with {\tt\small yes} answers have had size $1$, otherwise the queue would have not empty.
Therefore, $|\gee| = cnt < \tau$, meaning that $\gee$ is uncovered.
\end{proof}
}

\techrep{
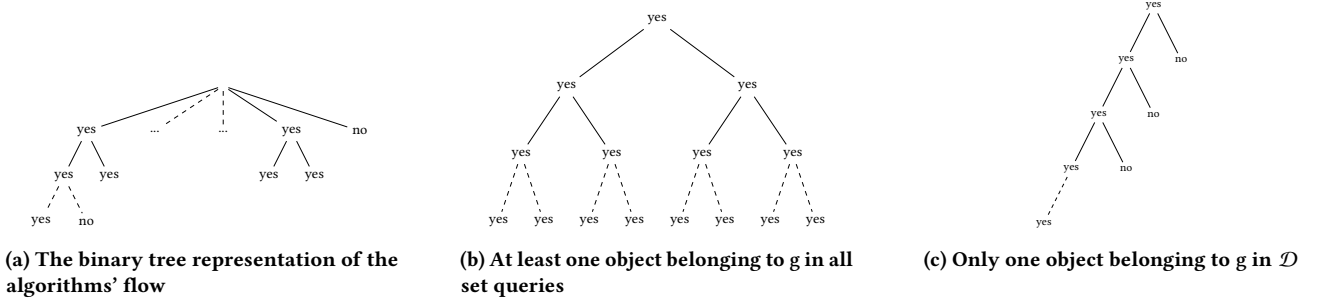
\begin{figure*}[!ht]
    \begin{subfigure}[t]{0.3\textwidth}
    \centering
        \resizebox{.9\linewidth}{!}{
\begin{tikzpicture}[level distance=1cm,
  level 1/.style={sibling distance=1.5cm},
  level 2/.style={sibling distance=1cm},
  level 3/.style={sibling distance=1cm}]
  \node {   }
    child {node {yes}
        child {node {yes}
            child[dashed] {node {yes}}
            child[dashed] {node {no}}
        }
        child {node {yes}
        }
        }
    child[dashed] {node {...}}
    child[dashed] {node {...}}
    child {node {yes}
        child {node {yes}
        }
        child {node {yes}
        }
        }
    child {node {no}
    };
\end{tikzpicture}
}
            \caption{The binary tree representation of the algorithms' flow}
            \label{fig:randomTree}
    \end{subfigure}
    \hfill
    \begin{subfigure}[t]{0.3\textwidth}
    \centering
        \resizebox{.9\linewidth}{!}{
\begin{tikzpicture}[
            level 1/.style={sibling distance=4cm},
            level 2/.style={sibling distance=2cm},
            level 3/.style={sibling distance=1cm}]
            \node {yes}
                child {node {yes}
                    child {node {yes}
                        child[dashed] {node {yes}}
                        child[dashed] {node {yes}}
                    }
                    child {node {yes}
                        child[dashed] {node {yes}}
                        child[dashed] {node {yes}}
                    }
                    }
                child {node {yes}
                    child {node {yes}
                        child[dashed] {node {yes}}
                        child[dashed] {node {yes}}
                    }
                    child {node {yes}
                        child[dashed] {node {yes}}
                        child[dashed] {node {yes}}
                    }
                };
        \end{tikzpicture}
}
        
            \caption{At least one object belonging to $\gee$ in all set queries}
            \label{fig:yes_case}
    \end{subfigure}
        \hfill
    \begin{subfigure}[t]{0.3\textwidth}
    \centering
        \resizebox{!}{.6\linewidth}{
\begin{tikzpicture}[level distance=1.5cm,
            level 1/.style={sibling distance=1.5cm},
            level 2/.style={sibling distance=1.5cm}]
            \node {yes}
                child {node {yes}
                    child {node {yes}
                        child {node {yes}
                                child[dashed] {node {yes}}
                                child[missing]
                            }
                        child {node {no}}
                    }
                    child {node {no}}
                }
                child {node {no}};
        \end{tikzpicture}
        }

        
            \caption{Only one object belonging to $\gee$ in $\dee$}
            \label{fig:no_case}
    \end{subfigure}
    \vspace{-4mm}\caption{Various outcomes of the binary tree of the generated tasks}
    \vspace{-5mm}
\end{figure*}
\begin{figure}[!bt]
        \centering
 

\resizebox{\linewidth}{!}{
\begin{tikzpicture}[
            level 1/.style={sibling distance=4cm},
            level 2/.style={sibling distance=2cm},
            level 3/.style={sibling distance=1cm}]
            \node { $\Box \Box \Box \Box \bigtriangleup \Box \Box \bigtriangleup \Box \Box \Box \Box \bigtriangleup \bigtriangleup \Box \bigtriangleup$ }
                child {node[align=center] {$\Box \Box \Box \Box \bigtriangleup \Box \Box \bigtriangleup$\\yes}
                    child {node[align=center] {$\Box \Box \Box \Box$\\no}
                    }
                    child {node[align=center] {$\bigtriangleup \Box \Box \bigtriangleup$\\yes}
                        child {node[align=center] {$\bigtriangleup \Box$\\yes}}
                        child {node[align=center] {$\Box \bigtriangleup$\\yes}}
                    }
                    }
                child {node[align=center] {$\Box \Box \Box \Box \bigtriangleup \bigtriangleup \Box \bigtriangleup$\\yes}
                    child {node[align=center] {$\Box \Box \Box \Box$\\no}
                    }
                    child {node[align=center] {$\bigtriangleup \bigtriangleup \Box \bigtriangleup$\\yes}
                    }
                };
        \end{tikzpicture}
}
        
        \vspace{-2mm}
        \caption{\color{\blue}\bc: Running Example}
\end{figure}
}

\submit{
\begin{figure*}[!ht]
   \begin{minipage}[t]{0.78\linewidth}
        \centering
    \begin{subfigure}[t]{0.325\textwidth}
    \centering
        \resizebox{.9\linewidth}{!}{
\begin{tikzpicture}[level distance=1cm,
  level 1/.style={sibling distance=1.5cm},
  level 2/.style={sibling distance=1cm},
  level 3/.style={sibling distance=1cm}]
  \node {   }
    child {node {yes}
        child {node {yes}
            child[dashed] {node {yes}}
            child[dashed] {node {no}}
        }
        child {node {yes}
        }
        }
    child[dashed] {node {...}}
    child[dashed] {node {...}}
    child {node {yes}
        child {node {yes}
        }
        child {node {yes}
        }
        }
    child {node {no}
    };
\end{tikzpicture}
}
            \caption{The binary tree representation of the algorithms' flow}
            \label{fig:randomTree}
    \end{subfigure}
    \hfill
    \begin{subfigure}[t]{0.325\textwidth}
    \centering
        \resizebox{.9\linewidth}{!}{
\begin{tikzpicture}[
            level 1/.style={sibling distance=4cm},
            level 2/.style={sibling distance=2cm},
            level 3/.style={sibling distance=1cm}]
            \node {yes}
                child {node {yes}
                    child {node {yes}
                        child[dashed] {node {yes}}
                        child[dashed] {node {yes}}
                    }
                    child {node {yes}
                        child[dashed] {node {yes}}
                        child[dashed] {node {yes}}
                    }
                    }
                child {node {yes}
                    child {node {yes}
                        child[dashed] {node {yes}}
                        child[dashed] {node {yes}}
                    }
                    child {node {yes}
                        child[dashed] {node {yes}}
                        child[dashed] {node {yes}}
                    }
                };
        \end{tikzpicture}
}
        
            \caption{At least one object belonging to $\gee$ in all set queries}
            \label{fig:yes_case}
    \end{subfigure}
        \hfill
    \begin{subfigure}[t]{0.325\textwidth}
    \centering
        \resizebox{!}{.6\linewidth}{
\begin{tikzpicture}[level distance=1.5cm,
            level 1/.style={sibling distance=1.5cm},
            level 2/.style={sibling distance=1.5cm}]
            \node {yes}
                child {node {yes}
                    child {node {yes}
                        child {node {yes}
                                child[dashed] {node {yes}}
                                child[missing]
                            }
                        child {node {no}}
                    }
                    child {node {no}}
                }
                child {node {no}};
        \end{tikzpicture}
        }

        
            \caption{Only one object belonging to $\gee$}
            \label{fig:no_case}
    \end{subfigure}
    \vspace{-4mm}\caption{Various outcomes of the binary tree of the generated tasks}
    \end{minipage}
    \hfill
    \begin{minipage}[t]{0.21\linewidth}
        \centering
 

\resizebox{\linewidth}{!}{
\begin{tikzpicture}[
            level 1/.style={sibling distance=4cm},
            level 2/.style={sibling distance=2cm},
            level 3/.style={sibling distance=1cm}]
            \node { $\Box \Box \Box \Box \bigtriangleup \Box \Box \bigtriangleup \Box \Box \Box \Box \bigtriangleup \bigtriangleup \Box \bigtriangleup$ }
                child {node[align=center] {$\Box \Box \Box \Box \bigtriangleup \Box \Box \bigtriangleup$\\yes}
                    child {node[align=center] {$\Box \Box \Box \Box$\\no}
                    }
                    child {node[align=center] {$\bigtriangleup \Box \Box \bigtriangleup$\\yes}
                        child {node[align=center] {$\bigtriangleup \Box$\\yes}}
                        child {node[align=center] {$\Box \bigtriangleup$\\yes}}
                    }
                    }
                child {node[align=center] {$\Box \Box \Box \Box \bigtriangleup \bigtriangleup \Box \bigtriangleup$\\yes}
                    child {node[align=center] {$\Box \Box \Box \Box$\\no}
                    }
                    child {node[align=center] {$\bigtriangleup \bigtriangleup \Box \bigtriangleup$\\yes}
                    }
                };
        \end{tikzpicture}
}
        
        \vspace{-2mm}
        \caption{\color{\blue}\bc: Running Example}\label{fig:runningexample}
    \end{minipage}
    \vspace{-5mm}
\end{figure*}
}

After correctness, let us now study the number of tasks \bc generate.
Before studying the performance in general cases, let us consider two extreme cases, while assuming $N=n$, where (Case I) the answer to all questions is yes, and (Case II) there exists only one object in $\gee$.

\noindent{\bf Case I:} 
Consider the cases where the answer to all set questions is yes, meaning that all set queries contain at least one object belonging to $\gee$. In this case, the set queries do not help prune the search space. As a result, the execution tree is a complete binary tree shown in Figure \ref{fig:yes_case}. Note that the number of leaf nodes in this tree shows the value of $cnt$, knowing that the answers to those set queries are {\tt\small yes}.
On the other hand, the \bc algorithm stops when $cnt = \tau$.
Therefore, the number of leaf nodes in the tree is at most $\tau$. In a complete binary tree with $\tau$ leaves, there are $\tau-1$ non-leaf nodes. As a result, the total number of tasks generated by the \bc algorithm, in this case, is $2\tau-1 = \Theta(\tau)$.

\noindent{\bf Case II:} 
Suppose there exists only one object from $\gee$ in $\dee$.
In this case (Figure \ref{fig:no_case}), every level of the execution tree contains exactly one node while one of its children is {\tt\small no}, while the other one is {\tt\small yes}.
Given that the root of the tree corresponds to a set of $n$ objects, the leaf {\tt\small yes} node (containing one object) is at depth $\log n$ in the tree.
Therefore, since every intermediate node has exactly two children, the number of tasks generated in this case is $\Theta(\log n)$.

\begin{theorem} \label{th:sblog}
Assuming the dataset size is $N=n$, or \textcolor{\blue}{when there is no limit on the set query size}, (i) the maximum number of tasks generated by the {\small \bc} algorithm is $\Theta(\tau\log n)$, (ii) the upper bound is tight.
\end{theorem}

\techrep{
\begin{proof}
We first prove the upper bound.
The algorithm prunes the sets associated with {\tt\small no} answers.
Let us assign each {\tt\small no} leaf node to its immediate non-leaf parent in the execution tree.
As a result, the response to the root node itself is {\tt\small no}, the number of {\tt\small no} leaf nodes is bounded by the number of non-leaf nodes in the tree.
On the other hand, given that the tree does not extend beneath the {\tt\small no} nodes, each non-leaf node is the path from a {\tt\small yes} leaf node to the root.
Hence, the set of non-leaf nodes can be computed as the union of the paths from all {\tt\small yes} leaf nodes to the root.
The number of {\tt\small yes} leaf-node is bounded by $\tau$.
Besides, as observed in Case II (Figure~\ref{fig:no_case}) a path from the root of each leaf node is of length at most $\log n$.
Therefore, using the {\em union bound}~\cite{motwani1995randomized}-Chapter 3.1, the total number of nodes in the execution tree is at most $\Theta(\tau\log n)$.

We use an adversarial example to show that the bound is tight.
Consider a case where $\gee$ is uncovered, there are exactly $\tau-1$ objects belonging to $\gee$, and those are uniformly distributed across the dataset such that the answer to the first $2\tau-3$ set queries is {\tt\small yes}. This is similar to Figure~\ref{fig:yes_case} except that there are $\tau-1$ nodes in level $\log \tau$.
At this point, since corresponding sets for each of the $\tau-1$ node at level $\log \tau$ contain exactly one object belonging to $\tau$, the subtree beneath each node grows similar to Figure~\ref{fig:no_case} until sets of the {\tt\small yes} nodes are of size one.
The depth of the subtrees beneath each of the $\tau-1$ nodes in level $\log \tau$ is $\log n - log \tau = \log \frac{n}{\tau}$, each containing $\Theta(\log \frac{n}{\tau})$ nodes.
As a result, the total number of nodes in the execution tree is $\Theta(\tau \log \frac{n}{\tau})$ showing that the $\Theta(\tau\log n)$ upper bound is tight.
\end{proof}
}

\begin{lemma}\label{lem:costanalysis}\vspace{-2mm}
(Cost Analysis) The maximum number of tasks generated by the \bc algorithm is $\Theta(\frac{N}{n} + \tau \log{n})$.
\end{lemma}

\techrep{
\begin{proof}
Given the upper-bound size $n<N$, the algorithm needs to partition $\dee$ into multiple sets queries, each of size $n$, as in the first level of Figure~\ref{fig:randomTree}.
The algorithm stops as soon as it finds $\tau$ queries with {\tt\small yes} answers.
In the worst case, however, it may require asking all levels of the $\frac{N}{n}$ set queries at level 1.
Suppose the number of level-1 set queries with {\tt\small yes} answer is less than $\tau$.
In that case, let $r_1,\cdots,r_t$, where $t<\tau$, be the level-1 {\tt\small yes} queries -- each being the root of the sub-tree beneath them.
As argued in the proof of Theorem~\ref{th:sblog}, since (a) there always are at most $\tau$ leaf nodes with {\tt\small yes} answers, (b) the number of {\tt\small no} leaf nodes is bounded by the number of intermediate nodes, (c) the set of all intermediate node in all sub-trees can be represented as the union of the paths from each {\tt\small yes} leaf to its corresponding root, and (d) the length of each path is at most $\log n$,
the maximum number of nodes in all sub-trees is $\Theta(\tau\log n)$.
Therefore, the maximum size of the execution tree is $\Theta(\frac{N}{n}+\tau\log n)$.
\end{proof}
}

Before concluding this section we would like to note that lower-bound on the maximum number of set queries an algorithm need to issue for deciding if $\gee$ is covered is $\frac{N}{n}$. That simply is because 
$\frac{N}{n}$ queries are needed to include each object in $\dee$ in at least one set query. In cases where $\gee$ is uncovered, all objects should be queried to verify $\gee$ is indeed uncovered.
Comparing this lower-bound with the maximum number of tasks generated by the \bc algorithm, one can verify that \bc has only a small {\em additive} overhead of $\Theta(\tau\log n)$ from the lower-bound.

\vspace{-2.5mm}
\subsection{Coverage Identification using {\small\bc}}
\textcolor{\blue}{Recall that the algorithm presented so far considered the single binary attribute case.}
Particularly, given a group $\gee$, \bc efficiently interacts with the crowd to identify if $\gee$ is covered.
This algorithm can be applied for coverage identification for different scenarios of sensitive attributes, as explained in the following.

\vspace{-2mm}
\subsubsection{Multiple Non-intersectional Groups} In a case where there is one attribute-of-interest $x$, let $c$ be the cardinality of $x$. In this case, each value of $x$ specifies a group $\gee_i$. To
    study coverage in this case, one needs to identify if each of the groups $\gee_i$ is covered in $\dee$. This can simply be done by running the \bc algorithm $c$ times, where the $i$-th run checks if the group $\gee_i$ is covered. As a result, following Lemma~\ref{lem:costanalysis}, the maximum number of tasks generated for this case is then $Theta\big(c(\frac{N}{n} + \tau \log{n})\big)$.

\vspace{-2mm} 
\subsubsection{Intersectional Groups} 
For cases with intersectional groups, the intersections of attribute values on $\mathbf{x}=\{x_1,\cdots,x_d\}$ specify different subgroups, while the objective is to identify the uncovered region in form of {\em maximal uncovered patterns (MUPs)}~\cite{asudeh2019assessing}.
Consider the graph representation of patterns associated with the attributes $\mathbf{x}$. 
A real-life example of this case is represented with {\small \tt gender} and {\small \tt race} attributes in Figure \ref{p:graph}.
The first level of the graph contains patterns, such as {\tt\small female-X} or {\tt\small X-black} with one attribute-value specified for them.
A child of a node at level $\ell$ is a node at level $\ell+1$ with one more attribute-value specified than its parent(s). 
For example, the level-2 pattern {\tt\small female-black} (representing the subgroup of {\tt\small female black} individuals) is a child of patterns {\tt\small female-X} and {\tt\small X-black}.
An MUP is a pattern that is uncovered itself but all of its parents are covered.
For example, {\tt\small female-black} is an MUP if it is uncovered but both its parents {\tt\small female-X} and {\tt\small X-black} are covered.

In order to specify MUPs using the \bc algorithm, we consider the fully-specified subgroups (at maximum level), each being the intersection of $d$ attribute values. 
The Cartesian product of the values of $x_1$ to $x_d$ determine these subgroups.
For example, in Figure~\ref{p:graph} the nodes at level 2 of the graph (e.g. {\tt\small female-asian}) show the fully-specified subgroups.
Using $c_1,\cdots,c_d$ as the cardinality of the attributes $x_1,\cdots,x_d$, the number of fully-specified subgroups is $m = c_1\times c_2\times\cdots\times c_d$.
Let $\gee_1,\cdots,\gee_m$ be the set of these subgroups.
Finding the MUPs is then possible by combing the Pattern-combiner algorithm~\cite{asudeh2019assessing} with the algorithm \bc.
At a high level, the pattern-combiner starts from the bottom of the pattern graph (fully-specified subgroups), and counts the coverage for each of those, eliminating the covered nodes (with all of their parents).
Then moving up in the graph, for the nodes that have not been pruned, it uses the counts of their children to check if those are covered, pruning the covered nodes and their parents while identifying the MUPs (uncovered nodes that all of their parents are pruned).

Algorithm \ref{alg1:coverageBinary} enables running pattern-combiner noting that \bc finds the {\em exact count} for an uncovered group. In such cases the
lower-bound variable ($cnt$ in Algorithm~\ref{alg1:coverageBinary}) contains the number of items belonging to the uncovered group.
Additionally, pattern-combiner only needs the counts for the fully-specified subgroups that are uncovered as it already prunes the covered nodes. 

\vspace{-4mm}
\section{Practical Optimizations based on Group Aggregation}\label{s:agg}
\vspace{-2mm}

So far in previous sections, we designed an efficient algorithm to detect if a certain group is covered in $\dee$. We further explained how this algorithm enables coverage identification for different cases of attributes of interest.
In this section, we propose practical heuristics for cases with multiple non-intersectional and intersectional groups.
In particular, we note that independently running the \bc algorithm for different groups specified by $\mathbf{x}$ to identify coverage, misses the opportunity to reuse the information collected during each run.

\submit{
\begin{figure}[!tb]
\centering
    \includegraphics[width=\columnwidth]{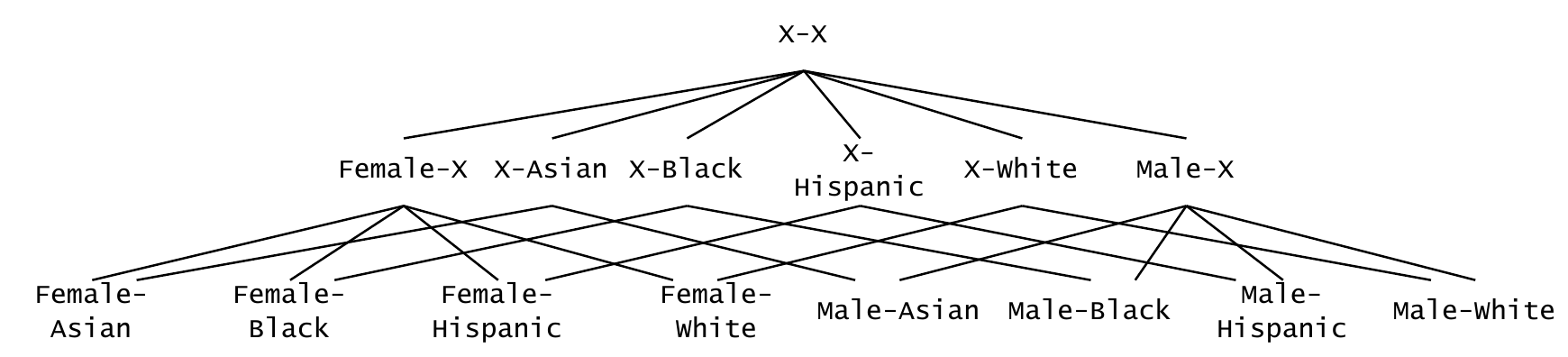}
\vspace{-7mm}\caption{The pattern graph for $\mathbf{x_1}$={\small \tt gender} and $\mathbf{x_2}$={\small \tt race}}
\label{p:graph}
\vspace{-6mm}
\end{figure}
}
\techrep{
\begin{figure*}[!tb]
\centering
\resizebox {1.8\columnwidth} {!} {
    \includegraphics[width=\columnwidth]{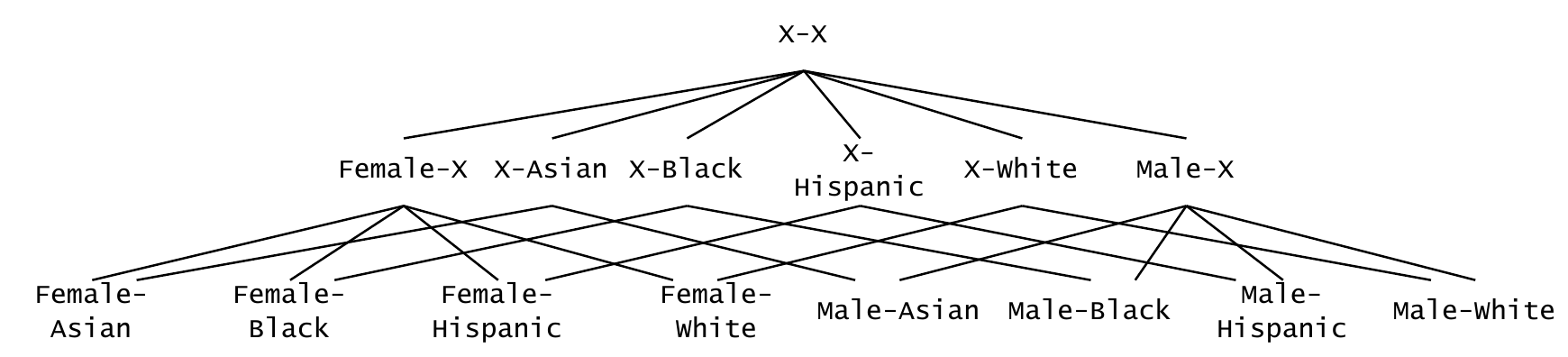}
}
\caption{The pattern graph for $\mathbf{x_1}$={\small \tt gender} and $\mathbf{x_2}$={\small \tt race}}
\label{p:graph}
\end{figure*}
}

First,  to avoid labeling objects multiple times we move the labeled objects from the unlabeled set $\dee$ to the labeled set $\el$.
Our main idea, however, is to form set queries that combine multiple demographic groups in one task -- instead of one group.
Specifically, consider a case where two or more demographic groups in the dataset are uncovered and there are very few items corresponding to them such that after combining them to one ``{\em super-group}'', the result is still uncovered.
For example, suppose the races {\tt\small Native American}, {\tt\small Asian}, and {\tt\small Middle Eastern} are on the absolute minority that the summation of counts for all these three groups is less than the coverage threshold.
In this case, instead of running the \bc algorithm once for each of them, we can run it for the super-group once.
To do so, we change the set query to combine the individual groups with {\em OR} predicate. 

Nevertheless, the challenge is that we do not have prior knowledge about the dataset $\dee$ to form the super-groups.
\textcolor{\blue}{In order to obtain such information, we consider {\em estimating the counts using sampling}.}
To do so, 
we add a sampling phase 
at the beginning of our method, 
in which a small random subset of the dataset is presented to the crowd as {\em point queries}, with their task being to label the items. 
The results from this step gives us an estimation of the demographic groups frequencies in the dataset. Based on this estimation, the algorithm will decide which groups to aggregate as super-groups. 
\textcolor{\blue}{Next, we need to determine the sample-set size.
Our intuition is that point queries are efficient for
verifying the coverage of the majority group, since we expect to discover enough of those after $\Theta(\tau)$ point queries.
Therefore, we can first issue the point queries to identify the majority group, while at the same time, we 
}
{\em piggyback} on the point query results to collect information about minorities and form the super-groups.
Following this idea, we consider labeling a random subset of size $c\tau$ of $\dee$ at the beginning of the algorithm, where $c$ is a small constant (we found $c=2$ as a good choice in our experiments).
Note that in cases where initial point queries do not find at least $\tau$ objects from the majority group(s), the algorithm effectively identifies a subset of these instances and needs fewer queries to get to the coverage threshold. As a result, \bc rapidly stops detecting them as covered.


One drawback of forming the super-groups is when the result for a set of super-groups is {\em covered}. In this case, we could not know whether one, two, or all groups are covered, and thus, we need to examine each separately. 
In other words, the aggregation strategy will incur a penalty cost when the super-groups are covered.


\begin{algorithm}[!tb]
    \caption{\nbc}\label{alg:coverageSmul}
    \begin{algorithmic}[1] \small
    \Require{ Dataset $\dee$, dataset size $N$, subset size upper bound $n$, coverage threshold $\tau$, target groups $\g$, \textcolor{\blue}{sample-size parameter $c=2$}}
    \Ensure{Coverage of all groups in $\g$}
        \State $\dee, \el \gets$ {\sc LabelSamples} $(\dee,\tau,c)$ {\tt\small ~~//obtain $c\tau$ random labels}
        \textcolor{\blue}{
        \State $\g_{agg} \gets${\sc Aggregate}$(\el,\tau,\g)$ {\tt\small //form the super-group}
        \State $cov\gets$ empty set
        }
        \For{$\Gee \in \g_{agg}$}:
            \State $\tau' \gets \tau-\sum_{\gee\in\Gee}\el.${\sc count}$(\gee)$
            \State \textcolor{\blue}{$cvg, cnt \gets $\bc $(\dee, N,n,\tau',\Gee)$}
            \State {\bf if} $|\Gee| = 1$ {\bf then} $cov$.{\sc add} \textcolor{\blue}{($\langle\Gee,cvg, cnt\rangle$)}; {\bf continue}
            \If{$cvg=$ {\bf true}}
            {\tt\small //if the super-group $\Gee$ is covered}
                \For{$\gee \in \Gee$}:
                    \State $\tau' \gets \tau-\el.${\sc count}$(\gee)$
                    \State $cvg, cnt \gets$ \bc $(\dee, N,n,\tau',\gee)$
                    \State $cov$.{\sc add} ($\langle \gee,cvg,cnt\rangle$)
                \EndFor
        \EndIf
        \State {\bf else} {\bf for} $\gee\in\Gee$ {\bf do:} $cov$.{\sc add} ($\langle\gee,\mbox{\bf false},cnt\rangle$)
        \EndFor
    \State {\bf return} $cov$
    \end{algorithmic}
\end{algorithm}

\color{\blue}
\vspace{-2mm}
\stitle{{\sc aggregate} function}\footnote{The pseudo codes of the functions are provided in the \submit{technical report~\cite{techrep}}\techrep{appendix}.} 
(Line 2 of Algorithm~\ref{alg:coverageSmul}) Let $\g$ be the list of groups in one attribute or the set of fully-specified subgroups in the intersection of multiple attributes.
The count estimations based on the samples collected in the labeled set $\el$ are utilized to set up the super-groups for $\g$.
We calculate the expected number of instances corresponding to each group in the dataset based on their occurrence in the sample. 
Let $\el.\mbox{\sc count}(\gee)$ return the number of objects in group $\gee$ that belong to $\el$. 
Since point queries are selected randomly, $\el$ is a random sample of $\dee$. Hence, the expected size of $\gee$ is $\EX[|\gee|] = N(\el.\mbox{\sc count}(\gee))/|\el|$.
If the expected number is less than the coverage threshold, it is likely that this particular group is uncovered in the dataset and vice versa. 
Similarly, if the summation of the expected numbers for a set of groups is still less than $\tau$, the super-group formed by merging them is expected to be uncovered.
To use this idea for forming the super-groups, we first sort the groups based on their count values in $\el$ ascending. This helps to put the minority groups nearby and merge them as super-groups.
Then we makes a pass over the sorted groups
while maintaining the sum over their expected coverage.
So far as the expected sums are less than $\tau$, the algorithm keeps merging the groups into a super-group, and then it moves to the next super-group.
It finally returns the list $\g_{agg}$ of the super-groups.
\color{black}



\vspace{-2mm}
\stitle{Multiple and Intersectional Groups Coverage}\label{s:single}
Using the idea of merging minority groups into super-groups, Algorithm~\ref{alg:coverageSmul} specifies the uncovered groups for the cases where there exist multiple, non-intersectional groups in a single attribute.
In particular, for every group $\Gee$ in the set of aggregated groups $\g_{agg}$, the algorithm first specifies the number $\tau'$ of instances it needs to observe before it can conclude $\Gee$ is covered.
Next it runs the \bc algorithm for identifying the coverage of $\Gee$ in $\dee$.
If $\Gee$ is not a super-group, the algorithm directly adds the coverage result of $\Gee$ to the output.
For cases where $\Gee$ is a super-group, if $\Gee$ is covered, the algorithm fails to conclude if groups $\gee\in\Gee$ are covered or not. Therefore, it reruns the \bc algorithm for all of such individual groups $\gee$.
On the other hand, for cases where the super-group $\Gee$ is uncovered, the algorithm concludes 
that all groups in $\Gee$ are uncovered.

We can take advantage of the above technique for multiple attributes, where we are interested in identifying the coverage of each of the individual and the intersectional groups. Figure \ref{p:graph} demonstrates the pattern graph for two attributes of {\small \tt gender} and {\small \tt race}. To solve the problem for this case, we take on a similar idea to the {\sc Pattern-Combiner} algorithm\cite{asudeh2019assessing}. The objective of the {\sc Pattern-Combiner} algorithm is to find MUPs (maximal uncovered patterns) in a dataset. As mentioned before, an MUP is a pattern that is uncovered but all of its parents are covered. Consequently, all of the children of MUP are uncovered as well. For example, in the {\small \tt race} and {\small \tt gender} attributes case with $\tau=50$, assuming that we find 15 instances of {\small \tt Asian-Female} and 20 instances of {\small \tt Asian-Male}, we can conclude that {\small \tt Asian} group with total instances of 35 is uncovered as well.
On the other hand, \textcolor{\blue}{if there were 28 {\small \tt Asian-Female} and 32 instances of {\small \tt Asian-Male}, we could conclude that {\small \tt Asian} group is covered, without any additional crowdsourced tasks.}

We use this idea to reduce the problem of identifying the coverage of multiple attributes to identifying the coverage of the fully-specified subgroups at the maximum level. We can see that this problem can be easily transformed into solving it for multiple, non-intersectional groups. It is noteworthy that the aggregation process for this special case requires that only the nodes with the same parent be aggregated with each other. To this end, we used a flag ({\em multi}) in our aggregation algorithm to distinguish between the two cases.
Having identified the coverage of the subgroups using the \nbc algorithm, we then proceed to identify the coverage of all other patterns in the upper levels. Algorithm \ref{alg:coverageMutli} describes the details of the discussed method.
\begin{algorithm}[!tb]
    \caption{\mc} \label{alg:coverageMutli}
    \begin{algorithmic}[1] \small
    \Require{Dataset $\dee$, dataset size $N$, subset size upper bound $n$, coverage threshold $\tau$, set of attributes $x$}
    \Ensure{ Coverage for all individual and intersecting groups in $x$}
        \State $\el \gets${\sc LabelSamples}$(\dee,\tau)$
        \State Let $\g$ be the set of fully-specified sub-groups at the max level
        \State \textcolor{\blue}{$\g_{agg} \gets${\sc Aggregate}$(\el,\tau,\g, multi=true)$}
        \State $cov \gets$\nbc$(\dee, N,n,\tau,\g_{agg})$
        \State Let $Q$ = an empty queue
        \State {\bf for} \textcolor{\blue}{$\langle \gee, cvg, cnt \rangle \in cov$} {\bf do} $Q.${\sc add}$(\langle \gee, cvg, cnt \rangle)$
        \While{$Q$ is not empty}
            \State $T\gets Q.del\_top()$
            \If{$T.cvg=$true}
                \State \textcolor{\blue}{{\bf foreach} $p$ in $T$.ancestors {\bf do:} $Q.${\sc add}$(\langle p,$true$, T.cnt \rangle)$}
            \Else
                \State $cnt \gets 0$
                \For{$\forall P \in T.parent.children - T$}
                    \State $cnt\gets cnt+P.cnt$; $Q.pop(P)$
                \EndFor
                \State {\bf if} $cnt \ge \tau$ {\bf then} $Q.${\sc add}$(\langle T.parent,$true$, T.cnt \rangle)$
                \State {\bf else} $Q.${\sc add}$(\langle T.parent,$false$, T.cnt \rangle)$
            \EndIf
        \EndWhile
    \end{algorithmic}
\end{algorithm}

\vspace{-3mm}
\section{Utilizing Existing Predictors}\label{sec:classifier}

\textcolor{\blue}{
While as discussed in \S~\ref{sec:intro}, {\em solely} employing ML models for coverage identification might be a problematic,
}
in presence of accurate and well-developed models, we should be able to utilize them in order to reduce the coverage identification cost -- i.e., the number of crowdsourcing tasks.
In this section, we adjust our core algorithm for this purpose. In such settings, instead of calling \bc in subsequent algorithms, one should call our classifier-aware algorithm (\css \xspace -- Algorithm~\ref{alg:classifier}).

\textcolor{\blue}{
Using pre-trained classifiers on the dataset gives us the prediction groups.
We would like to use the predicted labels in order to reduce the prediction cost.
}
{\em We still need to validate the \underline{correctness}} of the results obtained by the classifier in order to determine the coverage of a given group. 
Imagine an example in which a gender classifier applied to a dataset, classifying a set $\mathbf{f}$ as females. In order to identify the coverage for the female group, the main idea is to eliminate the falsely identified females (false positives), namely males, from $\mathbf{f}$. 
To this end, we apply a similar idea to what we did in \bc algorithm: we create a crowdsourcing task with all the points in the identified females set $\mathbf{f}$, and ask a {\em reverse question}: ``{\tt\small Is there any individual in this set that \underline{is NOT female}?}''. 
If the answer is yes, 
it means there are some false positives in this set. Therefore,
we take the divide-and-conquer approach
by dividing the set into two halves and repeating the question
until all false positive instances are eliminated. 


A performance issue with this strategy, however, happens when the false positive rate of the classifier for the given group is high. In such settings the divide and conquer strategy keeps dividing the sets into fine granularity, resulting in many small set queries to ask.
In such cases, labeling the data points in the female set using point queries to verify the classifier's label might be more efficient.

Following this observation, we propose a sampling phase to estimate the precision of the classifier on the positive group , \textcolor{\blue}{i.e., the group on which we would like to verify coverage} (females in this example). Similar to our proposed method in \S \ref{s:agg}, we choose a small, random sample from the identified females set (in our experiments, we found that a sample size of 10\% of the set classified as the given group would be a good choice). In the next step, we ask the crowd to label the sample using point queries and estimate the precision of the classifier, comparing the classifier label and the true label. We experimentally found that if 25\% of the sample are false positives, it is safe to say that the precision of the classifier on the positive group is sub-optimal for the group coverage identification task.
Based on the estimated precision of the classifier on the target group, our algorithm decides whether to eliminate the false positive objects using either partitioning or labeling strategies. 

At the end of this process, we will have a set containing only objects associated with the queried group $\gee$ ({\small \tt female} in this example). If we already have at least $\tau$ instances of a group $\gee$, we can determine the coverage and stop the process. 
However, assuming that the number of discovered instances $c'$ is less than $\tau$, now we have to find at least $\tau-c'$, {\em false negative} instances of $\gee$, i.e., the instances that belong to $\gee$ but are classified as not belonging to $\gee$ before we can conclude $\gee$ is covered (or to verify the number of false positive is less than $\tau-c'$ and hence $\gee$ is uncovered).
This can be done by applying the algorithm \bc on the set of objects classified as not $\gee$ with the threshold $\tau-c'$.
The details of the proposed method can be found in Algorithm \ref{alg:classifier}.

Our experiments in \S~\ref{exp:pred}, utilizing various pre-trained classifiers show that our optimization for coverage identification for classified datasets can achieve remarkable performance.


\begin{algorithm}[!tb]
    \caption{\css} \label{alg:classifier}
    \begin{algorithmic}[1] \small
    \Require{Dataset $\dee$, dataset size $N$, subset size upper bound $n$, coverage threshold $\tau$, target group $\gee$}
    \Ensure{ Coverage of $\gee$}
        \State \textcolor{\blue}{Let $G$ be the set of tuples in $\dee$ with predicted label $\gee$}
        \State \textcolor{\blue}{Let $S$ be a sample of $G$}
        \State \textcolor{\blue}{{\bf for} $t\in S$ {\bf do} $\el$.add$(\langle t, ${\sc PointQuery}$(t)\rangle)$}
        \State {\bf if} $\el.${\sc count}$(\gee) > 0.25 |G|$ {\bf then} $\Gee \gets${\sc Partition}$(\dee, G, n)$
        \State {\bf else} $\Gee \gets${\sc Label}$(\dee, G, \tau)$
        \State {\bf if} $|\Gee| \ge \tau$ {\bf then} {\bf return} true
        \State {\bf else} {\bf return} \bc$(\dee-G,N,n,\tau-|G|,\gee)$
    \end{algorithmic}
\end{algorithm}

\vspace{-2mm}
\section{Experiments}

In this section, we evaluate the performance of our proposed methods for coverage identification of single and multiple demographic groups. Additionally, we explore our heuristic of coverage detection for a single demographic group ({\small \tt gender}) on pre-processed image data with pre-trained classifiers. Finally, we deploy our system on a real crowdsourcing platform, Amazon Mechanical Turk, to explore the performance of the discussed methods with real workers.

\vspace{-3mm}
\subsection{Experiment Setup}
The experiments were conducted using both synthetic and real datasets. The algorithms were implemented in Python.
\begin{itemize}[leftmargin=*]
    \item \textbf{Synthetic datasets:} To thoroughly assess our algorithms, we create synthetic data with a variety of distributions. As an example, in a single demographic group problem setting for {\small \tt gender}, a data point can be either {\small \tt \{'F', 'M'\}}. In these experiments, we simulate the behavior of the crowdworkers in answering queries.
    \item \textbf{Image datasets:} We use image datasets for the purpose of experiments on MTurk and applying pre-trained classifiers on the data. In this section, we used slices of FERET DB \cite{feretdb} and UTKFace \cite{zhifei2017cvpr}.
    {\em FERET DB} is a dataset of 14,126 images of 1199 individuals taken from a variety of angles. {\em UTKFace} is also an image dataset consisting of over 20,000 annotated facial images.
\end{itemize}

\noindent{\bf Evaluation Plan}:
We evaluate the performance of {\small \bc} and the optimizations for multiple non-intersectional and intersectional groups as well as classified data using pre-trained classifiers. We report the number of tasks required for each experiment setting.
\textcolor{\blue}{
We use a straightforward baseline, called {\small \sc Base-Coverage}, to identify group coverage as a baseline method to compare our algorithm's results to. In this method, each task is created containing only one single data point and is asked about.} Two outcomes are possible for this algorithm: 
either at some point, $\tau$ instances of objects associated with the group are identified and hence, the group is covered, or the algorithm goes through all the data and determines that the group is  uncovered.

\noindent{\bf Default Values}:
To evaluate the performance of our algorithms, we fix the value of some of the introduced parameters in \S \ref{s:binary}-\ref{sec:classifier}. The default value of $c$ is 2, and we fix $\tau$ and $n$ as 50 in all experiments except their respective parameter analysis experiments.

\vspace{-3mm}
\subsection{Summary of Results}

Our proposed \bc algorithm achieved remarkable results in group coverage identification in our experiments. Even at the worst possible case (where the number of instances associated with the group in the dataset is close to the coverage threshold), both the synthetic and MTurk experiments with the real crowdworkers showed that our algorithm needs a significantly small number of tasks compared to the size of the dataset in order to achieve results. We also show that the upper bound discussed in \S~\ref{single_analysis} is in fact tight.
Our optimizations on multiple group cases (both intersectional and non-intersectional) proved to be effective in most cases compared to a brute force approach utilizing the \bc algorithm to identify the coverage of multiple groups. Additionally, the optimization of classified data using pre-trained classifiers achieved notable results in most cases, decreasing the number of required tasks by approximately 80\% in some.

\vspace{-3mm}
\subsection{Proof of Concept}

\subsubsection{Amazon Mechanical Turk} \label{exp:amt}

In this experiment, we evaluate our proposed method and the crowd's performance on a live crowdsourcing platform. We defined a {\small \tt female} coverage identification task and published our HITs on Amazon Mechanical Turk. 
Each HIT contained a set of initial $n=50$ images to present a reasonable workload to the crowd, with a maximum assignment to 3 workers for quality control. The layout of the HITs was designed as Figure \ref{fig:squery}. The crowd was asked to answer the questions with yes or no. 

\textcolor{\blue}{
To examine how each method affects the final outcomes, we use three quality control techniques: Majority Vote, Qualification Test, and Rating \cite{daniel2018quality}.
We adopt the off-the-shelf \textit{Majority Vote} as a group assessment to control the quality of outputs, in which we assigned the same HIT to 3 workers and took the majority vote as the truth. We also experiment with \textit{Rating} and \textit{Qualification Test} as two types of individual assessments to further verify that our workers have the appropriate skills for the tasks.
For the rating assessment, we measure the workers' performance using \textit{NumberHITsApproved} and \textit{PercentAssignmentsApproved} in MTurk and only allow workers who meet a certain criterion for approved HITs and assignments to perform our tasks.
Additionally, we designed a qualification test to verify the workers' competence before granting them access to our HITs as another method for quality control. This test has a similar layout to the original HITs, which also served the purpose of familiarizing the workers with the tasks.
One interesting observation that was made is that despite the relatively low number of assignments per HIT, only 1.36\% of the total 660 answers from the crowd in all experiments (in 220 HITs) were incorrect which did not affect the final result in each experiment run. Additionally, we did not detect a significant difference between the experiments with or without individual assessments which further supports our idea about the tasks being relatively easy and straight-forward for the crowd and fault-proof to an extent.
}


We employed the fixed price model as our pricing strategy. In our first set of experiments, each HIT price was set to $\$0.1$. In the next experiments, we decreased the reward for each HIT to $\$0.05$. Interestingly, this did not discourage the workers to accept and complete our tasks. Overall, we paid a total of \textcolor{\blue}{$\$44.1$ to the workers and $\$8.82$ to Amazon MTurk as service charges.}

We used two different subsets of the FERET dataset. With $\tau = 50$, the results from each experiment setting can be found in Table \ref{table:mturk} with a comparison to the baseline {\small \sc Base-Coverage} method and our theoretical upper bound ($\frac{N}{n}+\tau \log(n)$).

\begin{figure}[t]
    \centering
    \footnotesize
    \begin{tabular}{||P{0.36\linewidth} |P{0.14\linewidth}|P{0.14\linewidth}|P{0.14\linewidth}||}
    \hline
    {\bf FERET DB ({\small \tt females}=215, {\small \tt males}=1307)} & \bf \bc \#HITs & \bf {\sc base-coverage} \#HITs& \bf {\sc upper-bound} \#HITs  \\ \hline
      QC: Majority Vote  & 74 &342 & 115 \\ \hline
     \textcolor{\blue}{QC: Qualification Test, Majority Vote}  & \textcolor{\blue}{75} & \textcolor{\blue}{386} & \textcolor{\blue}{115} \\ \hline
     \textcolor{\blue}{QC: Rating (PercentAssignmentsApproved >= 95, NumberHITsApproved >= 100), Majority Vote} & \textcolor{\blue}{71} & \textcolor{\blue}{284} & \textcolor{\blue}{115} \\
     \hline
    \end{tabular}
    \captionof{table}{\color{\blue}Coverage identification for {\small \tt female} on Amazon MTurk}
    \label{table:mturk}
    \vspace{-8mm}
\end{figure}


\subsubsection{Existing predictors} \label{exp:pred}
To investigate the performance of the existing predictors on image datasets and our strategy to detect the coverage of a dataset utilizing these models, we ran a number of experiments using DeepFace \cite{serengil2020lightface} and another CNN-based facial demographic classifier \cite{facialCNN} to predict the gender of the individual in datasets. We used a subset of images of unique individuals from the FERET dataset \cite{feretdb}, and two 3,000-point subsets of UTKFace \cite{zhifei2017cvpr} with different distribution of females to evaluate the results for both covered and uncovered cases. We applied DeepFace with {\small \tt opencv} and {\small \tt retinaface} as the underlying face detectors. \textcolor{\blue}{Next, we passed the predicted labels and the females set detected by the classifier as the input to our \css algorithm. Our algorithm chooses between "partitioning" and "labeling" to eliminate false positives in the identified female set. It picks "partitioning" if the classifier is at least 25\% precise on a sample of the set, and "labeling" otherwise.} 
The results are reported in Table \ref{table:classifier}. \textcolor{\blue}{We also include the coverage detection results using \bc to compare them with those of \css.}

\begin{figure*}[!tb]
\centering
\footnotesize
\begin{tabular}{||P{0.23\linewidth} ||P{0.13\linewidth} |P{0.05\linewidth} |P{0.09\linewidth}||P{0.13\linewidth}|P{0.07\linewidth}|c||}
 \hline
 \bf & \multicolumn{3}{c||}{\bf \textcolor{\blue}{Existing classifiers' performance}} & \multicolumn{2}{c|}{\bf \textcolor{\blue}{\css}}  & \bf \textcolor{\blue}{\bc} \\
\bf dataset & \bf classifier & \bf accuracy & \bf precision on female group & \bf false positive elimination strategy & \bf \css \#HITs  & \bf \#HITs \\
\hline
\multirow{2}{*}{FERET DB ({\small \tt females}=403, {\small \tt males}=591)} & DeepFace (opencv) & 79.57 & 99.5 & Partition & 14 & 80 \\
& DeepFace (Retinaface) & 84.1 & 100.0 & Partition & 17 & 80\\
& BaseCNN & 64.48 & 59.19 & Label & 84 & 80\\
\hline
\multirow{2}{*}{UTKFace ({\small \tt females}=200, {\small \tt males}=2800)} & DeepFace (opencv) & 93.56 & 52.02 & Label & 97 & 51 \\
& DeepFace (retinaface) & 94.16 & 56.15 & Label & 89 & 51 \\
& BaseCNN & 97.6 & 74.8 & Label & 69 & 51 \\
\hline
\multirow{2}{*}{UTKFace ({\small \tt females}=20, {\small \tt males}=2980)} & DeepFace (opencv) & 96.53 & 8.0 & Label & 134 & 221 \\
& DeepFace (retinaface) & 96.43 & 10.09 & Label & 143 & 221\\
& BaseCNN & 97.6 & 21.59 & Label & 122 & 221 \\
\hline
\end{tabular}
\captionof{table}{The results of female group coverage detection on gender classified datasets}
\label{table:classifier}
\vspace{-9mm}
\end{figure*}

For each experiment, the accuracy and the precision of the classifier on the dataset are calculated separately. As discussed before, the accuracy of some classifiers can vary by a large margin on different data. Moreover, a high level of accuracy does not necessarily guarantee reasonable precision in predicting the class of data points. More specifically, both of the classifiers had a relatively high accuracy in the classification task on the UTKFace dataset, but both also had a low precision in their prediction for the female demographic group, which further proves the fact that the performance of the existing predictor for sensitive groups is questionable in many cases.

In addition, the results show that our heuristics make the right decision for which false positive elimination strategy to use in most cases, leading to significantly fewer necessary tasks to get the result. Compared to our \bc algorithm used standalone, the proposed techniques can produce significantly better results in most cases and still competitive results in others. 

\textcolor{\blue}{
Before concluding this section, we would like to remind that our algorithms for coverage multiple non-intersectional and intersectional sensitive attributes run the \bc algorithm in their core.
In other words, those algorithms can be viewed as issuing \bc multiple times. So we expect similar results for those cases, using the real crowd.
}

\textcolor{\blue}{
\subsection{Downstream Tasks Consequences}\label{exp:validation:downstream}
In this experiment, we show how the lack of coverage may cause model performance disparity (unfairness) in the downstream tasks. In particular, using two computer vision tasks we observe that (a) lack of coverage may cause model performance disparity for an uncovered group, and (b) resolving the lack of coverage reduces the performance disparity for the uncovered groups. 
\subsubsection{Drowsiness Detection}
Drowsiness detection systems are used to prevent accidents that are caused by drivers who fell asleep while driving. MRL eye dataset~\cite{Fusek2018433} is a large-scale human eye dataset containing infrared images captured in a variety of lighting conditions from 37 people. While some of the subjects in the dataset wear glasses, we intentionally disregarded the images of such subjects to make them uncovered and created a sample of size 26480 images belonging to two classes of {\em open} (14279 images) and {\em closed} eyes (12201 images). Following the same procedure, we generated 10 datasets and repeated each experiment 10 times, using different datasets. 
Using this as the training data, we built a CNN model, and evaluating the model, we observed that while it has an average overall accuracy of 91.5\%, the average accuracy for the spectacled subjects is only 81\%. Next, in order to confirm that the issue was due to the lack of coverage, we gradually added 20, 40, 60, 80, and 100 images from the uncovered region back to each of the classes of {\em open} and {\em closed} in the training data, and retrained and evaluated the model. The results are illustrated in Figure~\ref{fig:drowsiness}. With an increase in the number of samples taken from the uncovered group, we observed a reduction in the accuracy/loss disparity of the model between a randomly sampled test set and a sample consisting exclusively of spectacled subjects.   
\subsubsection{Gender Detection}
To further verify our proposal, we repeat a similar procedure, using UTKFace dataset. We extract a sample comprising 7055 face images from UTKFace such that each image belonged to a class of either {\em male} (3834 images) or {\em female} (3221 images). While extracting the sample, we intentionally picked the subjects only if they were Caucasian. Using this as the training data, we trained a CNN model to predict the gender of the subjects. We repeated this procedure on 10 different samples and observed that on average, there is a 1\% disparity in the overall accuracy of the model versus the accuracy for the Black subjects. Similar to the previous experiment, gradually increasing the number of Black subjects in the training data reduces the aforementioned disparity close to zero as seen in Figure~\ref{fig:gender}.
}

\begin{figure}[!htb] 
\centering
    \begin{subfigure}[t]{0.495\linewidth}
        \centering
        \includegraphics[width=\textwidth]{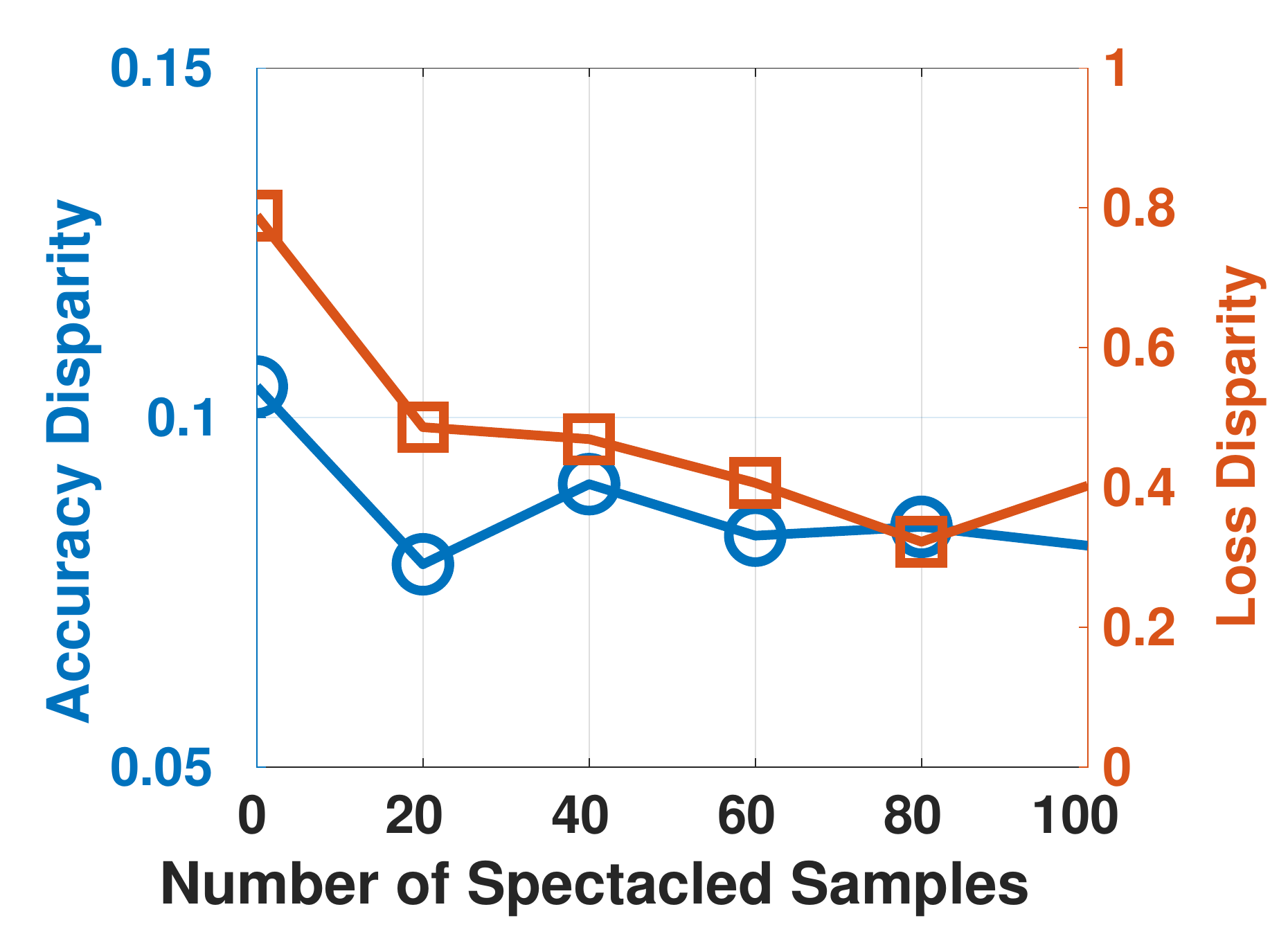}
        \vspace{-6mm}\caption{\color{\blue}drowsiness detection}
        \label{fig:drowsiness}
    \end{subfigure}
    \hfill
    \begin{subfigure}[t]{0.495\linewidth}
        \centering
        \includegraphics[width=\textwidth]{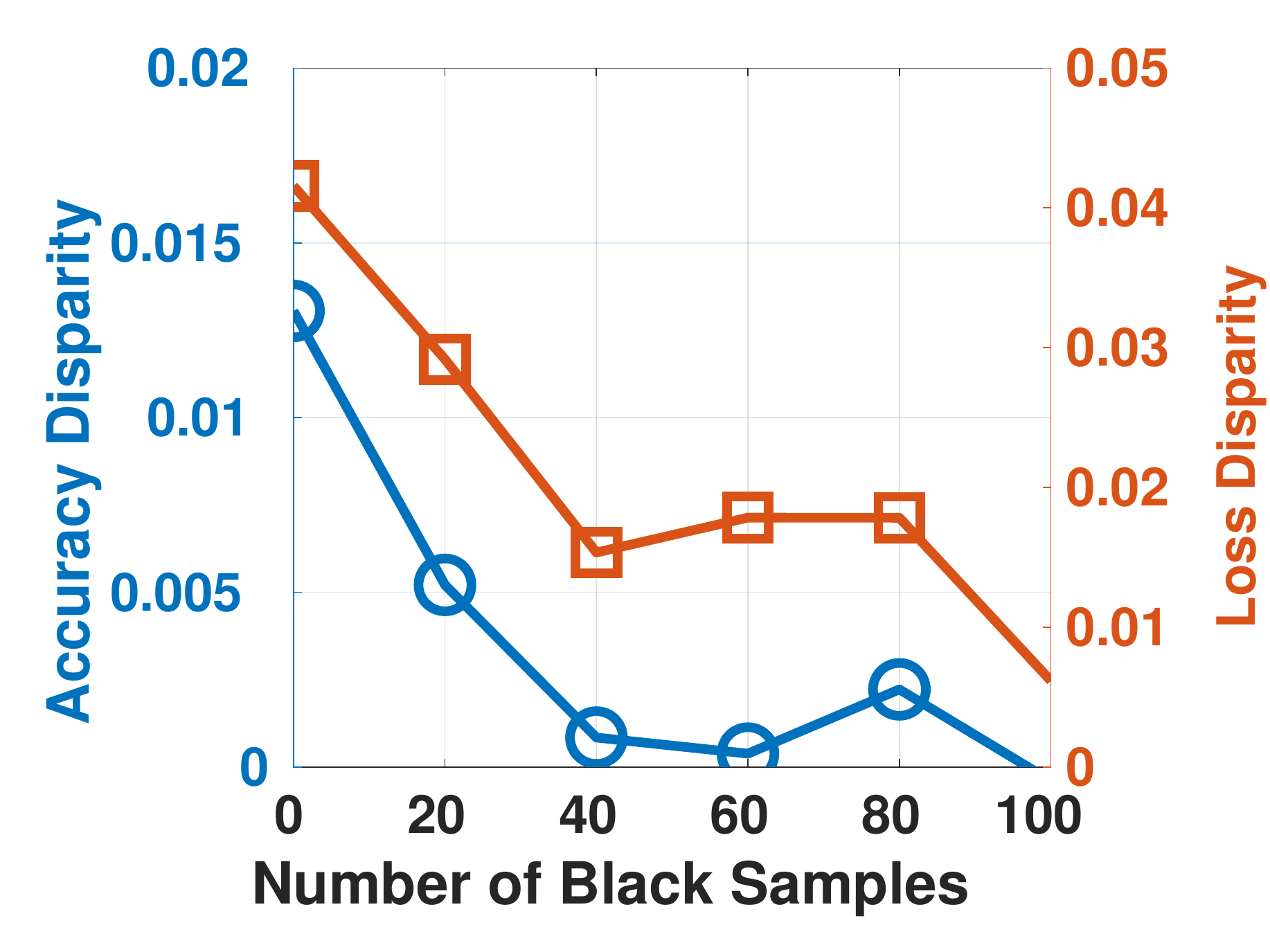}
        \vspace{-6mm}\caption{\color{\blue}gender detection}
        \label{fig:gender}
    \end{subfigure}
\vspace{-3mm}
\caption{\color{\blue}Effect of lack of coverage in the downstream tasks}
\vspace{-4mm}
\end{figure}

\subsection{Performance Evaluation}
In the following sections, we present the results of our experiments using the stated settings. First, we evaluate the performance of the \bc algorithm with varying parameters $\tau, n, N$ in \S \ref{ex:bin}. Next, we evaluate the optimizations for multiple non-intersectional and intersectional groups in \S \ref{exp:opt}.

\begin{figure*}[!tb]
\begin{minipage}[t]{\linewidth}
    \begin{subfigure}[t]{0.24\textwidth}
        \centering
        \includegraphics[width=\textwidth]{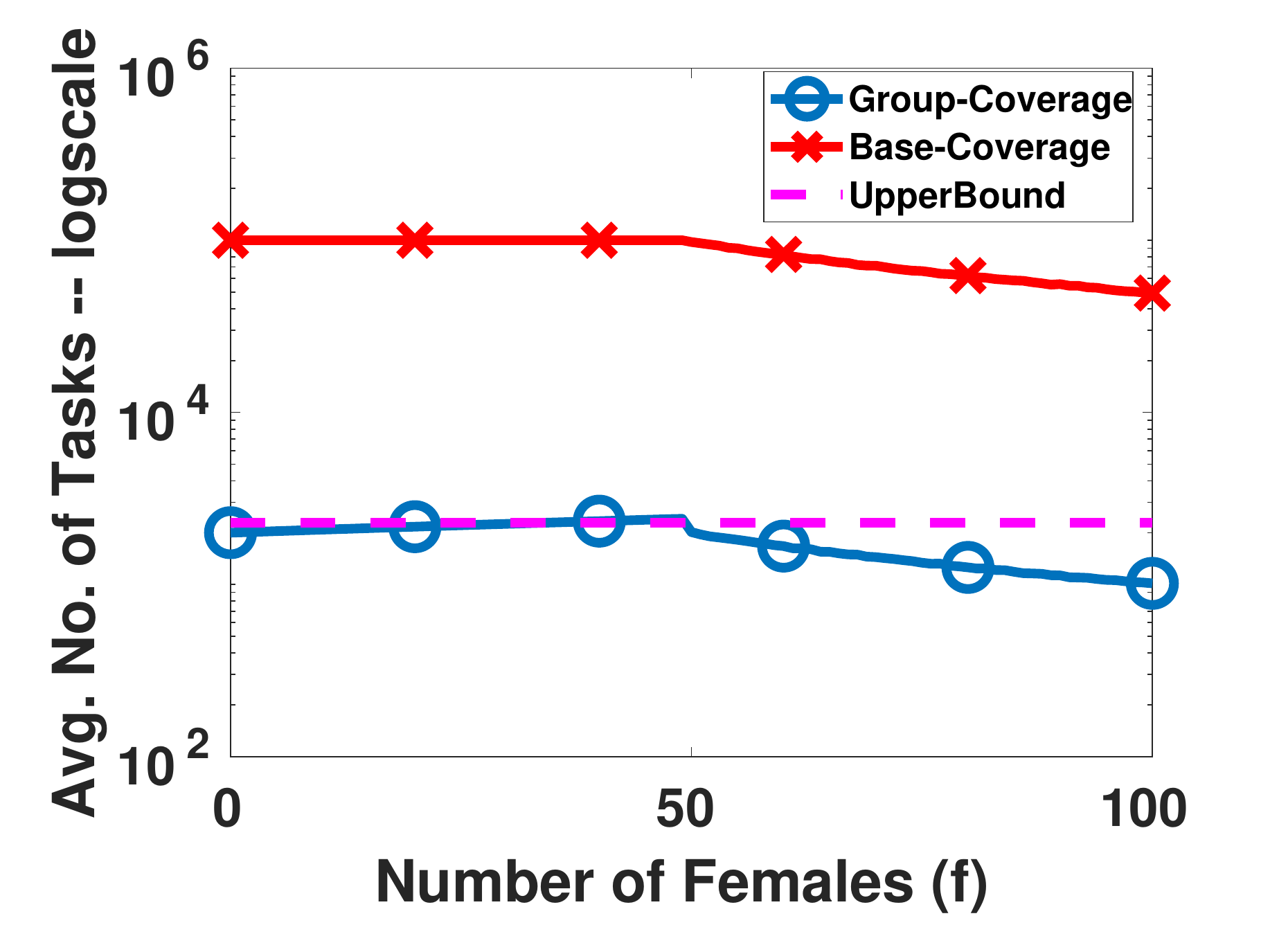}
        \vspace{-7mm}\caption{\color{\blue}varying \#females}
        \label{fig:coverageX}
    \end{subfigure}
    \hfill
    \begin{subfigure}[t]{0.24\textwidth}
        \centering
        \includegraphics[width=\textwidth]{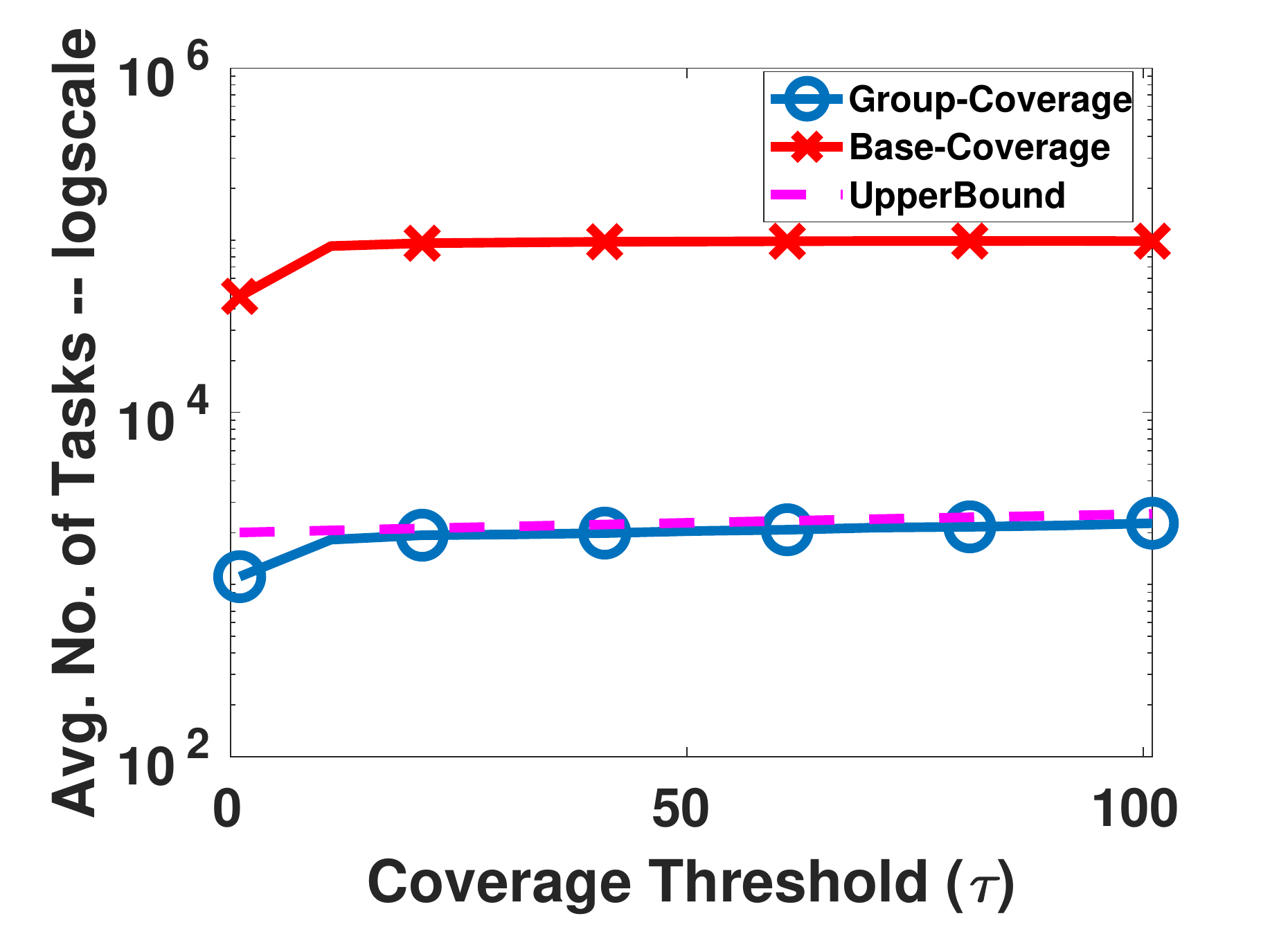}
        \vspace{-7mm}\caption{Varying coverage threshold}
        \label{fig:cov}
    \end{subfigure}
    \hfill
    \begin{subfigure}[t]{0.24\textwidth}
        \centering
        \includegraphics[width=\textwidth]{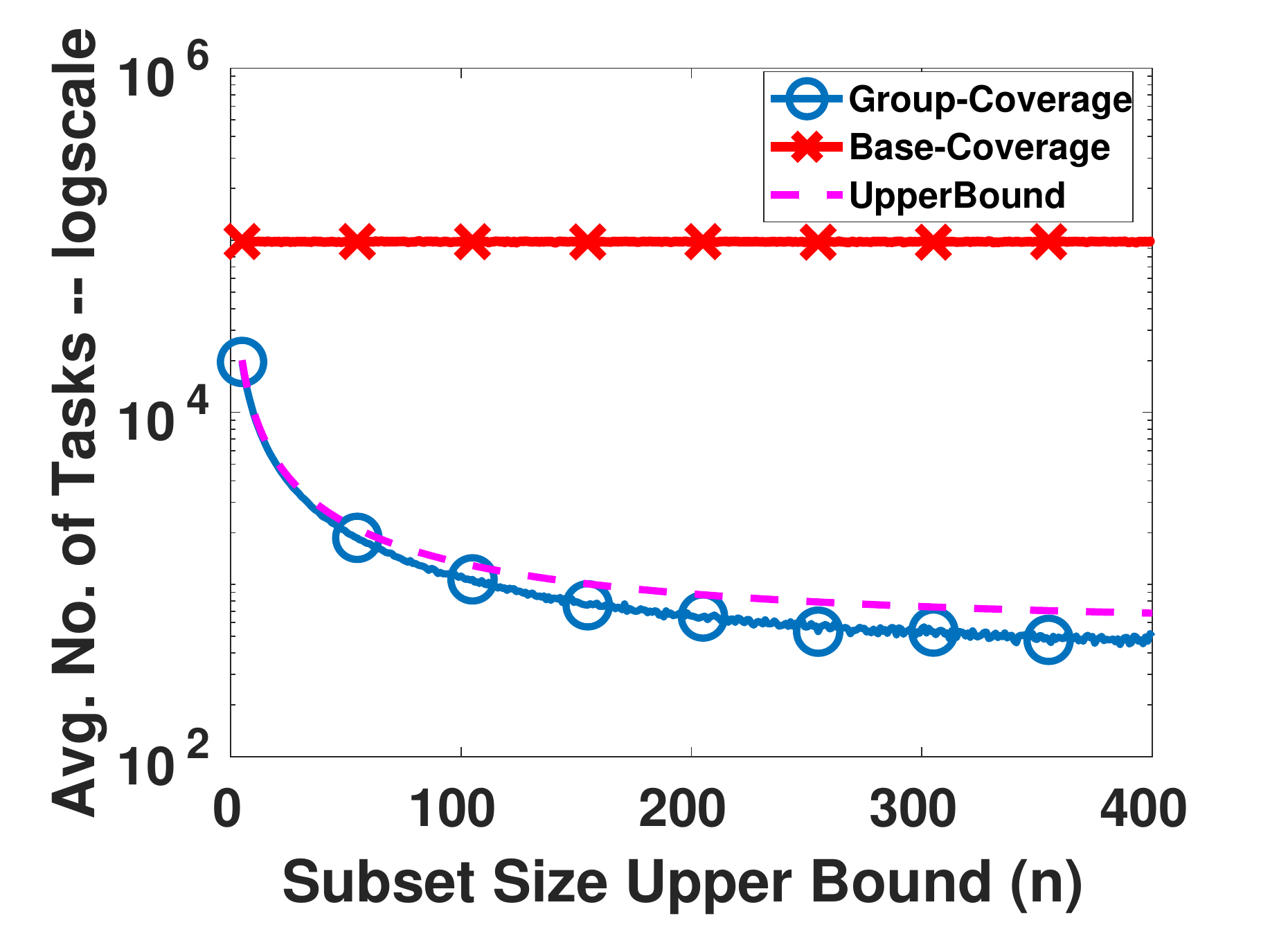}
        \vspace{-7mm}\caption{Varying subset size}
        \label{fig:subset}
    \end{subfigure}
    \hfill
    \begin{subfigure}[t]{0.24\textwidth}
        \centering
        \includegraphics[width=\textwidth]{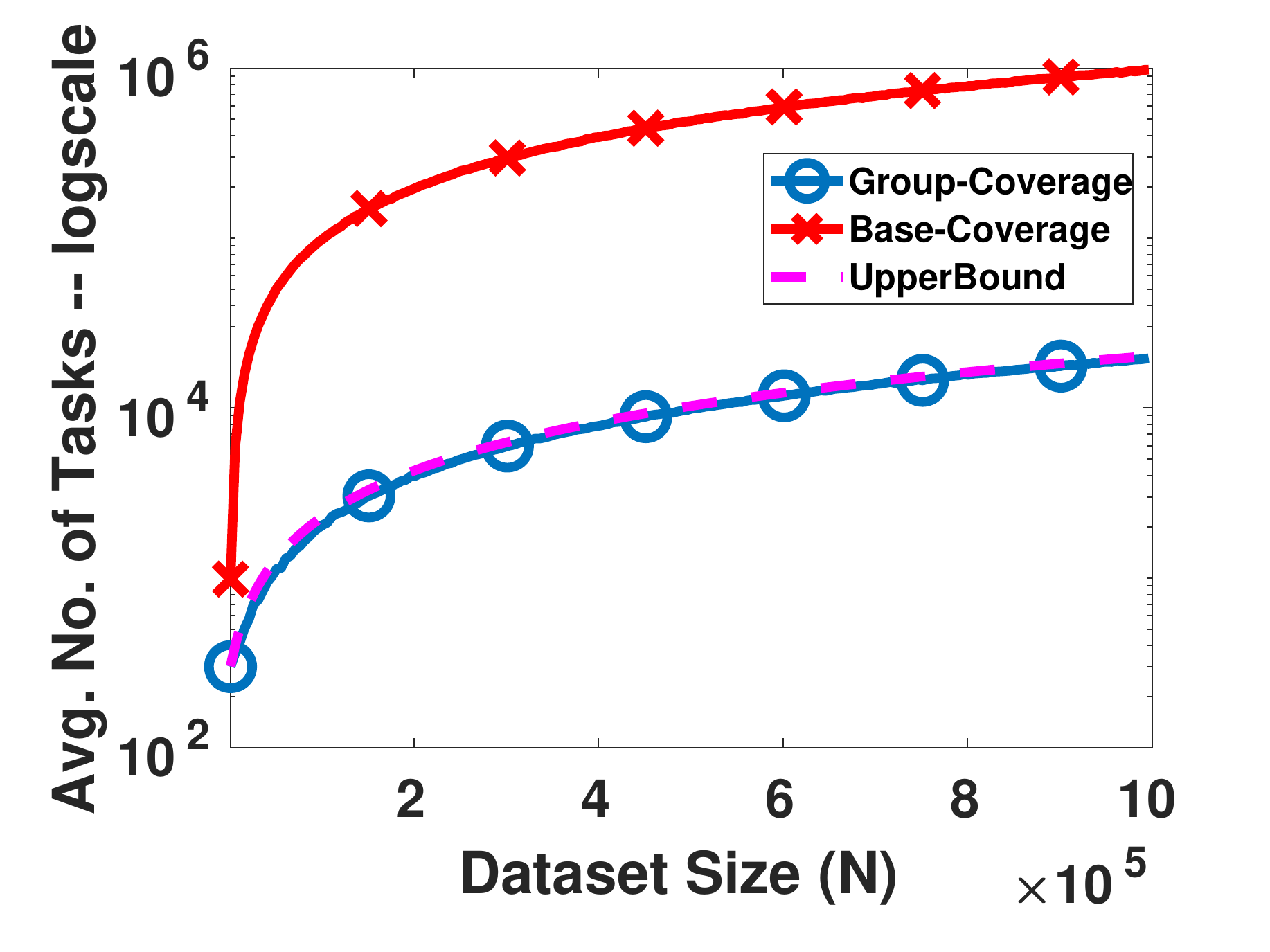}
        \vspace{-7mm}\caption{Varying dataset size}
        \label{fig:dataset}
    \end{subfigure}

    \begin{subfigure}[t]{0.24\textwidth}
        \centering
        \includegraphics[width=\textwidth]{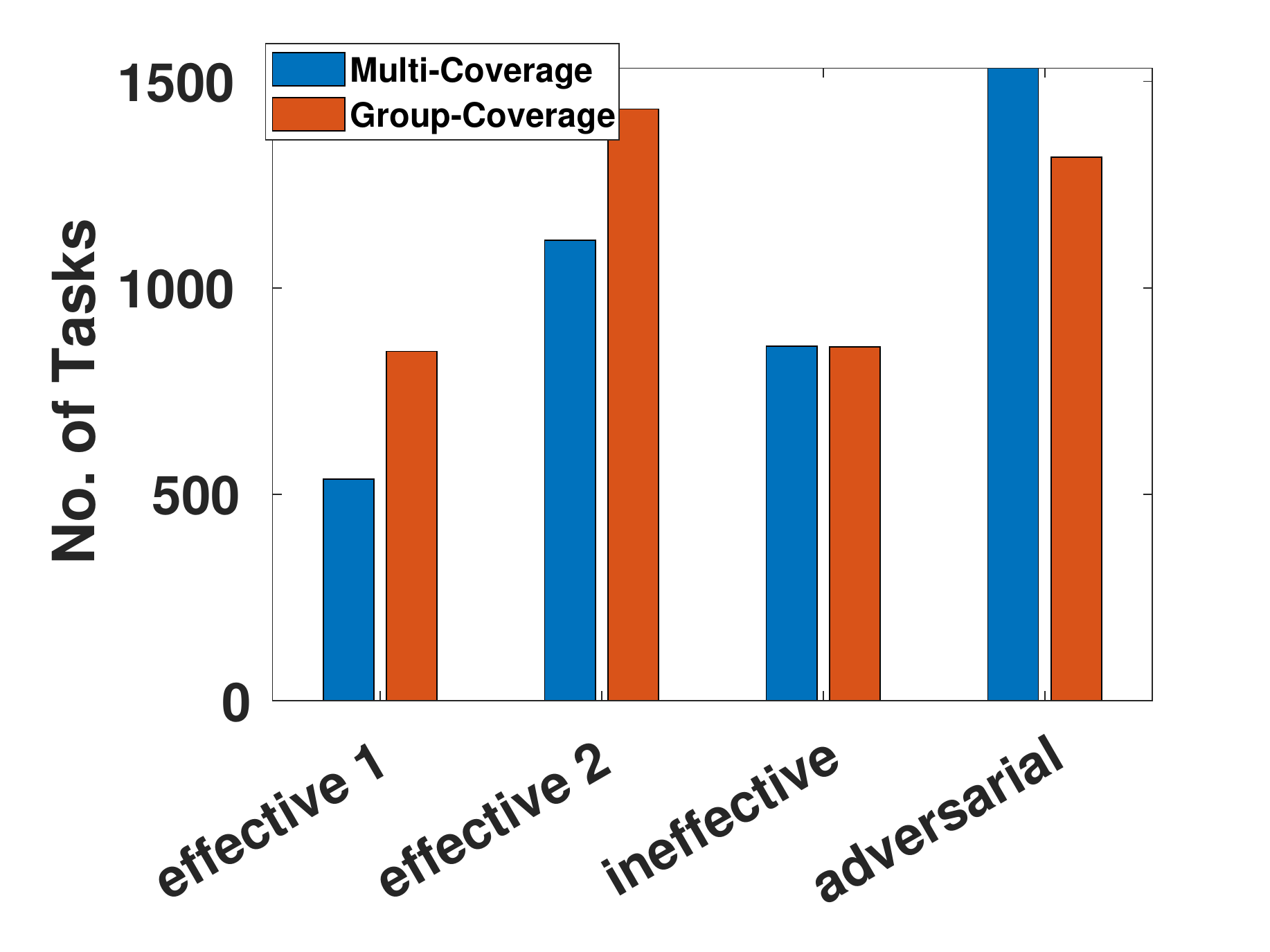}
        \vspace{-7mm}\caption{Multiple non-intersectional groups optimization ($\sigma=4$) vs. \bc}
        \label{fig:multi_bar}
    \end{subfigure}
    \hfill
    \begin{subfigure}[t]{0.24\textwidth}
        \centering
        \includegraphics[width=\textwidth]{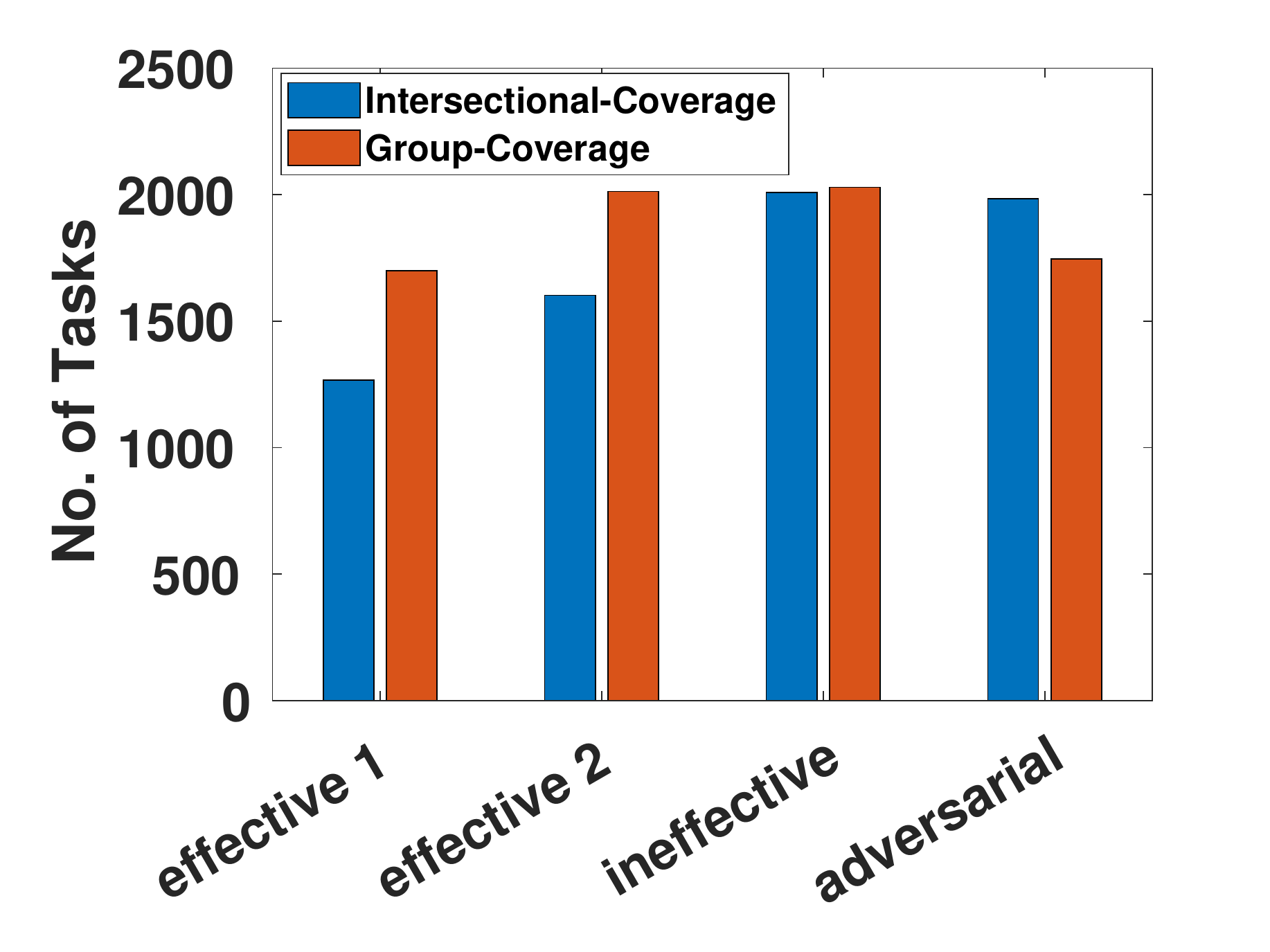}
        \vspace{-7mm}\caption{Intersectional groups optimization ($\sigma_1=2,\sigma_2=2,\sigma_3=2$) vs. \bc }
        \label{fig:inter_bar}
    \end{subfigure}
    \hfill
    \begin{subfigure}[t]{0.24\textwidth}
        \centering
        \includegraphics[width=\textwidth]{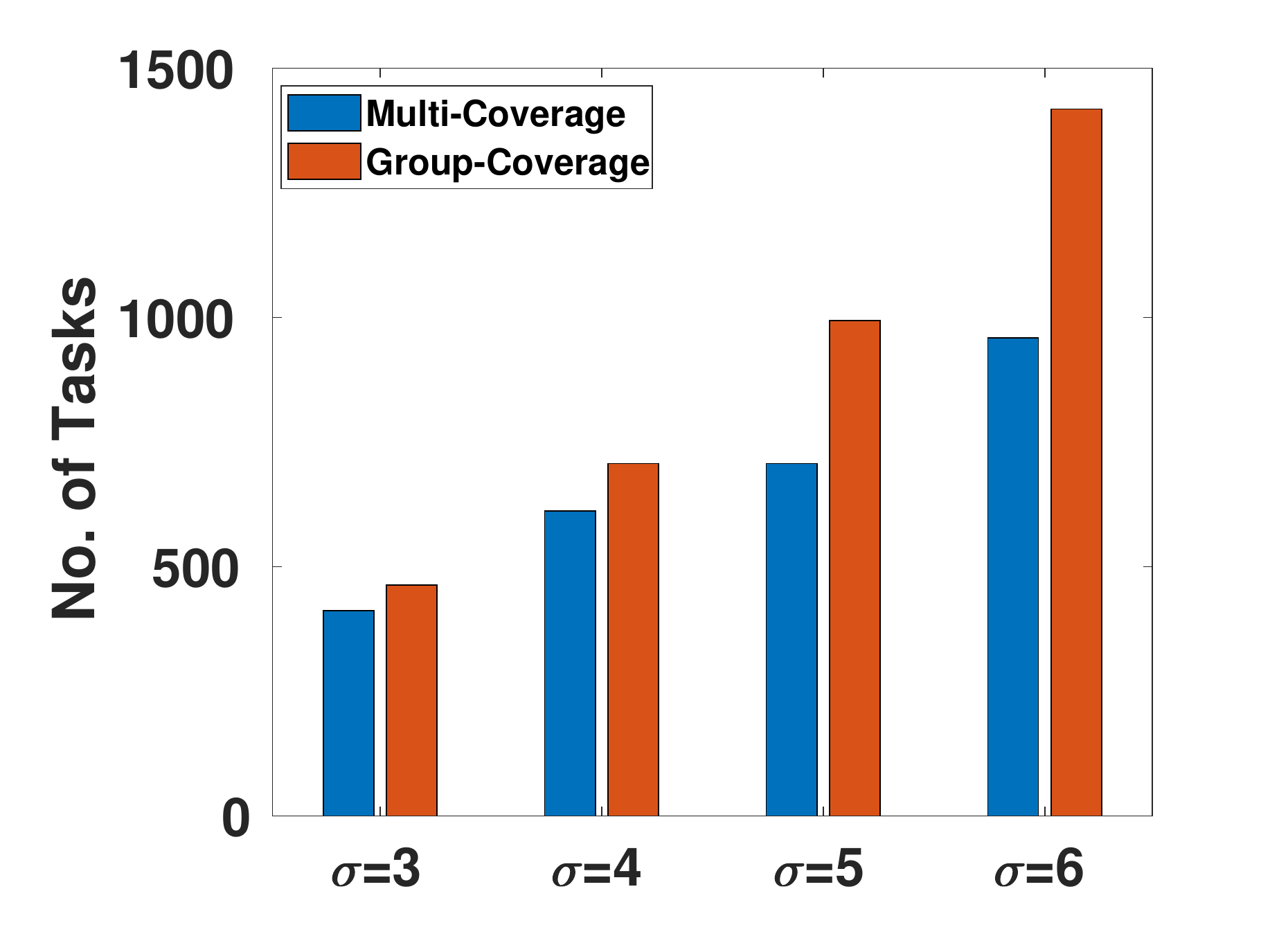}
        \vspace{-7mm}\caption{Multiple groups in one attribute with $\sigma=3,4,5,6$}
        \label{fig:nbOpt_2}
    \end{subfigure}
    \hfill
    \begin{subfigure}[t]{0.24\textwidth}
        \centering
        \includegraphics[width=\textwidth]{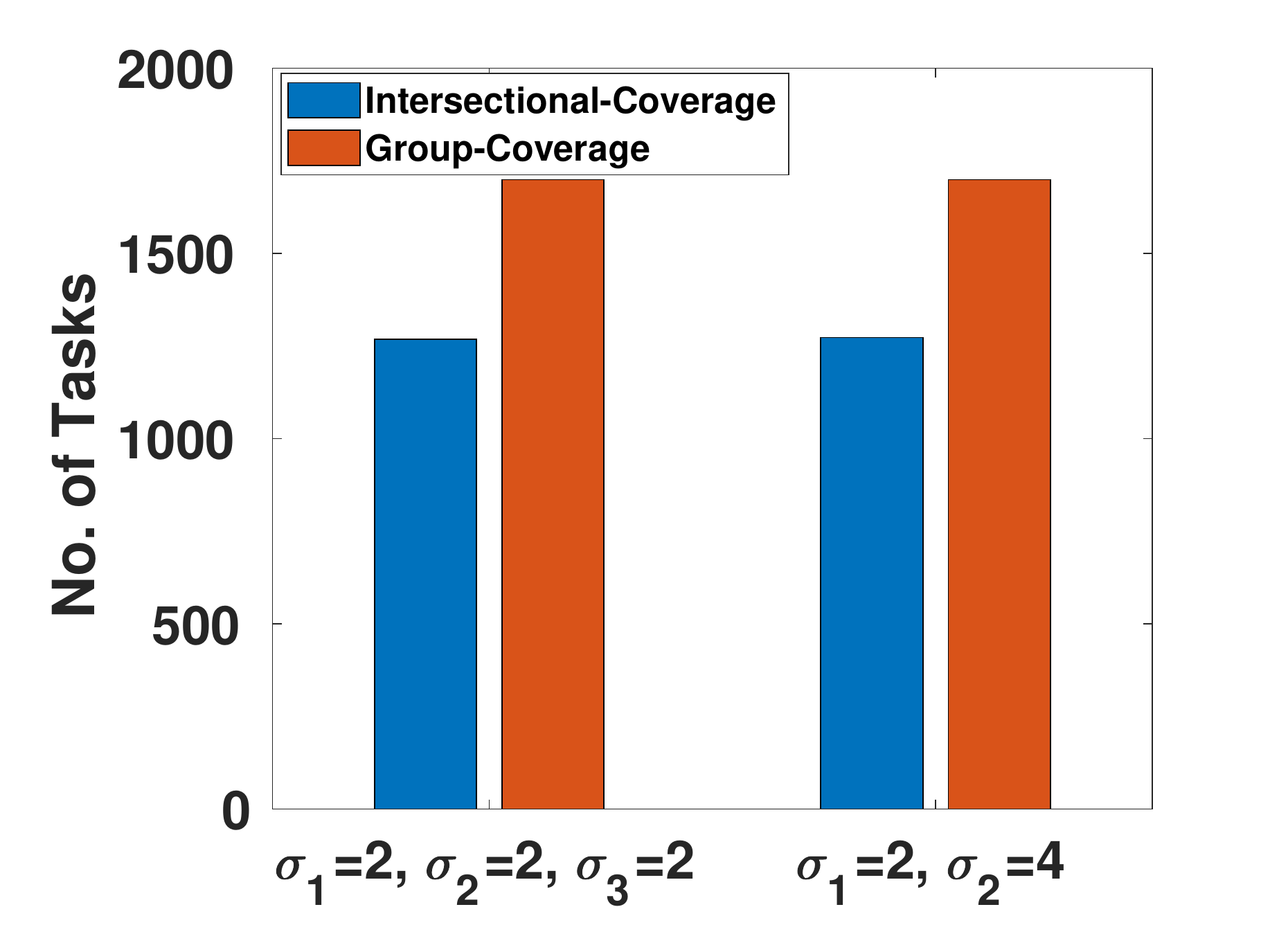}
        \vspace{-7mm}\caption{Multiple groups in two attributes with $\sigma_1=2,\sigma_2=4$ and $\sigma_1=2,\sigma_2=2,\sigma_3=2$}
        \label{fig:multiOpt_2}
    \end{subfigure}
    \vspace{-4mm}
    \caption{Performance evaluation for \bc, \nbc, and \mc algorithms}
\end{minipage}
\vspace{-5mm}
\end{figure*}

\subsubsection{\bc} \label{ex:bin}
To evaluate the performance of the \bc algorithm, we designed a simulation to reflect the procedure that the crowd would be presented with to carry out the tasks. The objective of the simulation is to first identify whether the dataset is covered with respect to a demographic group and determine the total number of tasks required to identify the coverage of a given group.
For this purpose, assuming we are interested in identifying the coverage of {\small \tt female}, we generate a dataset containing males and females and shuffle it randomly to prepare for the experiment. Each experiment with particular variables is run multiple times to better capture the effect of the dataset's underlying distribution on the results. In these sets of experiments, we study the impact of the scope of parameters on the end results.

\noindent{\bf Varying $\tau$}.
First, we analyze the relationship between the coverage threshold and the number of {\small \tt females} in the dataset and its impact on the number of necessary tasks to get the results. Figure \ref{fig:coverageX} illustrates the number of required tasks when there exist $[0,2\tau]$ items of the demographic group in the dataset. We have a dataset of size 100K, and we select the coverage threshold as $50$. It can be observed that the largest number of queries is needed when the number of {\small \tt females} ($f$) is close to $\tau$. Conversely, the farther $f$ gets from $\tau$, the fewer tasks are required to get to a conclusion. This observation is consistent with the discussion in \S \ref{single_analysis}; with too few or too large quantities of $f$ in the dataset, our algorithm's results appear to be further from the upper bound. 

Figure \ref{fig:cov} shows the results of running the algorithm with different coverage thresholds. The coverage threshold is varied from \textcolor{\blue}{1 to 100 (0.001\% to 0.1\% of the dataset size)} and there are exactly $\tau$ females at each run. Naturally, when the coverage threshold increases, the algorithm needs to cover more grounds to produce results. This also shows that the relationship between the coverage threshold and the cost is linear as discussed in \S \ref{single_analysis}. Note that the results in this figure demonstrate the case where $f = \tau$, which is the situation that requires the maximum number of tasks to get to a decision and is very close to the theoretical upper bound.

\noindent{\bf Varying $n$}.
This experiment is designed to study the impact of the subsets size upper bound on the algorithm's outcome. We set the coverage threshold to $50$ in a dataset of size $100,000$ while maintaining $50$ instances of females in the dataset. In Figure \ref{fig:subset}, we can see a substantial jump in the number of tasks when the subset size increases from around $10$ to $20$. Moreover, the result does not change significantly after that even with a notable increase in the subsets size. This confirms the logarithmic nature of the subsets size upper bound parameter in the algorithm. 

When determining the set size, one must consider the ability of the crowd to identify the subject of the task at hand. While selecting larger $n$ might lead to fewer required tasks, it is likely that we obtain less reliable answers from the crowd due to the large number of items presented all at once. Additionally, increasing the initial subset  size will not significantly impact on the end results.

\noindent{\bf Varying $N$}.
To assess the performance of our algorithm in a variety of datasets, we ran experiments for datasets of size 1K to 1M.

Figure \ref{fig:dataset} illustrates the results of the algorithm for varying dataset sizes. As expected, the number of required tasks to determine the results grows linearly with the size of the dataset, but never exceeds 6\%. In other words, our results show that in practice, we can determine the group coverage for a dataset with tasks no more than 6\% of the dataset size at the very worst case.

\vspace{-2.5mm}
\subsubsection{Optimizations for multiple groups} \label{exp:opt}
To evaluate the performance of our proposed method in identifying the coverage for multiple groups, we take a similar approach as the previous section. We create a synthetic dataset comprising of data points that can correspond to $\sigma=3,4,5,6$ distinct demographic groups for the non-intersectional case, and two datasets, one with 2 attributes with cardinalities $\sigma_1=2, \sigma_2=4$ and the other with 3 binary attributes ($\sigma_1=2, \sigma_2=2, \sigma_3=2$) for the intersectional case.
In a dataset of $10$K points, with a threshold of 50 and a subset size of 50, we vary the number of items for each group in the dataset to simulate different combinations of settings and run the algorithm for each variation. Additionally, we run the \bc algorithm for each group independently to compare our results. In our experiments, we found out that while our heuristics on multiple groups can perform very well in some cases, it can also appear to be ineffective or worse than the brute force in some other. The results of the experiments for these cases are shown in Figures \ref{fig:multi_bar} and \ref{fig:inter_bar}.

Each of the bars defined as effective 1, effective 2, ... represent a different setting with respect to the number of instances associated with each group which is further described in Table \ref{tab:exp_desc}.

The adversarial case, in which there are multiple uncovered groups with summation of items greater than $\tau$ in the dataset, it is likely that our heuristic fails in aggregating these groups into a super-group since the probability of having instances of these groups in the sample is significantly low. Thus, the super-group is covered and the algorithm needs to run for each of the subgroups individually. This imposes a penalty on the total number of tasks and makes it an adversarial case for our heuristic. To conclude, we can expect that our method works very well or with little difference compared to brute force in some cases while failing in others.
\begin{table}[t]
    \centering
    \footnotesize
    \begin{tblr}{||m{0.16\linewidth}|m{0.60\linewidth}||}
    \hline
        {\bf Setting} & {\bf Description} \\ \hline
        effective 1 &  3 uncovered minorities; their aggregated super-group is uncovered
        \\\hline
        effective 2 &  3 covered minorities\\ \hline
        ineffective & 2 uncovered and one covered minority\\ \hline
        adversarial & 3 uncovered minorities; their aggregated super-group is covered\\ \hline
    \end{tblr}
    \caption{Experiment settings for multiple groups}
    \label{tab:exp_desc}
    \vspace{-10mm}
\end{table}

Figure \ref{fig:nbOpt_2} shows the results for the \nbc algorithm for attributes with various cardinalities. Considering cases where our heuristic is effective, as the cardinality of the attribute increases, the total required tasks in \nbc grows more slowly than the brute force, resulting in a larger gap between the two methods as the cardinality increases.

Figure \ref{fig:multiOpt_2} represents the results of the \mc algorithm for two cases, one with 2 attributes with cardinalities $\sigma_1=2, \sigma_2=4$ and the other with 3 binary attributes ($\sigma_1=2, \sigma_2=2, \sigma_3=2$). As expected, with the same settings, the results for each of these cases are similar, with the number of fully-specified subgroups at the maximum level for both cases being equal. In other words, in the case of intersectional groups with multiple attributes, the only important feature is the cardinality of the attributes rather than the number of attributes.

\vspace{-2mm}
\section{Related Works}\label{sec:related}

\paragraph{Crowd-sourcing for Bias Detection}
\cite{hu2020crowdsourcing} proposes a crowd-sourcing workflow to facilitate sampling bias discovery in visual datasets with the help of human-in-the-loop. This workflow takes a visual data set as an input and outputs a list of potential biases of the data set. There are three steps in this workflow. The first step is \textit{Question Generation} and the crowd inspects random samples of images from the input dataset and describes their similarity using a question-answer pair. The next step is \textit{Answer Collection} in which the crowd review separates random samples of images from the input dataset and provides answers to questions solicited from the previous step. Finally, in the third step called \textit{Bias Judgement} the crowd judge whether statements of the visual dataset that are automatically generated using the questions and answers collected accurately reflect the real world. 

\vspace{-2mm}
\paragraph{Set queries}
\color{\blue}
Set-based HITs, similar to our set queries, have been used in various crowdsourcing studies, including crowd powered data mining~\cite{li2017crowdsourced,li2016crowdsourced}.
\cite{10.1145/2213836.2213878} first introduced the idea of filtering a set of data based on a particular property using humans. 
Another example is \cite{6228183}, where 
for the purpose of top-k and group-by queries, the crowd is asked to answer {\em type} set question which has ``yes'' or ``no'' answer based on whether the data points in a set have the same type, which is similar to our notion of set queries on a target demographic group.
Set queries have also been used in the crowd-sourced ``count'' operation: For instance, Marcus et al.~\cite{marcus2012counting} show a small batch of objects (images) to the crowd, asking them to estimate the number of items satisfying a specific constraint (e.g., photos with a car in them).
Set queries are also popular in crowd-sourced clustering. For example, in \cite{gomes2011crowdclustering} each worker views a small set of images as a HIT, where they are asked to provide a partial clustering of the set.
Set queries have also been used for tasks such as crowd-sourced median finding~\cite{heikinheimo2013crowd}, crowd-sourced planning~\cite{kaplan2013answering}, etc.


\vspace{-2mm}
\paragraph{Group testing}
Our proposed approach in the \bc algorithm falls under the general category of group testing approaches, where a task of identifying certain objects is broken up into tests on groups of items~\cite{du2000combinatorial}.
First proposed by Dorfman~\cite{dorfman1943detection}, group testing has been widely used across different domains~\cite{du2000combinatorial}, with early applications such as detecting broken electrical circuits~\cite{chen1989detecting} with more recent applications in graphs~\cite{cheraghchi2012graph}, web databases~\cite{asudeh2016query}, and even in Covid-19 detection~\cite{gollier2020group}.
Related work includes \cite{eppstein2007improved}, which explores ways to perform efficient combinatorial group testing to identify up to $d$ defective items from a set of $n$ items using a reduced number of tests for practical set sizes.
More generally, the class of divide and conquer (d\&c) algorithms are popular in the crowdsourcing. For example, \cite{zhang2011crowdsourcing} proposes a crowd-sourced d\&c approach for sorting. Similarly, \cite{negri2011divide} proposes a crowdsourcing d\&c approach for creating cross-lingual textual entailment corpora. Related work also includes crowdsourcing d\&c approaches for mobile platforms~\cite{amato2013divide}, paired comparisons~\cite{wang2023crowdc}, etc.

\color{black}
\vspace{-2mm}
\paragraph{Coverage} The notion of data {\it coverage} has been studied across different settings~\cite{shahbazi2023representation,jin2020mithracoverage,asudeh2019assessing,lin2020identifying,asudeh2021coverage,tae2021slice,accinelli2021impact,moskovitch2020countata,accinelli2020coverage}. 
With many angles to tackle, data coverage has been studied for datasets with discrete~\cite{asudeh2019assessing} and continuous~\cite{asudeh2021coverage} attributes populated in single or multiple \cite{lin2020identifying} relations. 
\techrep{
Additionally, \cite{accinelli2020coverage, accinelli2021impact,shetiya2022fairness}
consider resolving representation bias in preprocessing pipelines by rewriting queries into the closest operation so that certain subgroups are sufficiently represented in the downstream tasks. Similarly, \cite{nargesian2021tailoring} uses data integration as a mean to resolve representation bias.
}
\submit{Additionally, \cite{accinelli2020coverage, accinelli2021impact,shetiya2022fairness} (resp. \cite{nargesian2021tailoring}) use query rewriting (resp. data integration) to resolve representation bias.}
Existing works in data coverage have so far only focused on tabular data.

\submit{\vspace{-2mm}}
\section{Conclusion}
In this paper, we studied the problem of coverage identification in image data. This problem is motivated by the historical representation bias in various forms of data, and specifically the inefficiency of the existing supervised or unsupervised learning methods in performing equally well for minority groups on image data. We proposed an efficient algorithm to identify the coverage of a demographic group across the dataset and showed that the number of required tasks is optimal and close to the theoretical lower bound, and introduced practical heuristics to expand our solution for multiple non-intersectional or intersectional groups. We also presented an optimization method for detecting group coverage in datasets labeled by the existing, pre-trained predictors.

In this work, we focused on image data as a specific form of multimedia data. We hope to find equally efficient methods to identify data coverage in other forms of multimedia data such as video in our future work. In addition, our goal in this paper was mainly to minimize the cost of crowdsourcing by minimizing the total number of required tasks. We consider extending our techniques to support various pricing models as part of our future work. 

\section{Research Ethics Review Statement}
{\color{\blue}
The research conducted in this study involving participants from Amazon Mechanical Turk (MTurk) ensured adherence to ethical principles and guidelines. 
Participants were provided with clear and comprehensive information such as the purpose, procedures, the type of the task, the compensation amount, and the expected time to complete the task.
The participants were informed about the qualification screening.
Informed consent was obtained from all participants before their engagement in the study.
In addition to the platforms anonymization of MTurk workers, we further ensured the privacy of participants by not collecting personal information.
}

\submit{\newpage}
\bibliographystyle{ACM-Reference-Format}
\bibliography{bib}

\techrep{
\newpage
\section*{Appendix} \label{apx}
\appendix

\section{Pseudo-codes}

\begin{algorithm}[H]
    \caption{{\sc Partition} \& {\sc Label}} \label{alg:part}
    \begin{algorithmic}[1]
        \Function{\sc Partition}{$\dee,\Gee,n$}
        \State Let $Q$ = an empty queue
        \For{$i\gets0$ to $N$ with step size $n$}:
            \State $t \gets \{t_i,\cdots,t_j\}$
            \State $Q.add(t)$
        \EndFor
        \State Let $S$ = an empty set
        \While{$Q$ is not empty}
            \State $T\gets Q.del\_top()$
            \State $(i,j)\gets (T.b\_index,T.e\_index)$
            \State ans $\gets ${\sc AskQuestion}$(\{t_i,\cdots,t_j\},\gee')$
            \If{ans=no}
                \State $S.${\sc add}$(\{t_i,\cdots,t_j\})$
            \Else
                \If{$j>i$ {\tt\small /*if setsize>1*/}}
                    \State $T_1 \gets \{t_i,t_{\lfloor\frac{i+j}{2}\rfloor})$
                    \State $T_2 \gets (t_{\lfloor\frac{i+j}{2}\rfloor+1}, t_j)$
                    \State $Q.add(T_1)$; $Q.add(T_2)$
                \EndIf
            \EndIf
        \EndWhile
        \State {\bf return} $S$
        \EndFunction
        
        \Function{\sc Label}{$\dee,\Gee,\tau$}
        \State $cnt \gets 0$
            \For{$t \in \Gee$}
                \State $l \gets ${\sc PointQuery}$(t)$
                \State {\bf if} {$l \ne \gee$} {\bf then} $\Gee.${\sc Remove}$(t)$
                \State {\bf else} $cnt \gets cnt + 1$ 
                \If{$cnt \ge \tau$}   \State {\bf break} \EndIf
            \EndFor
            \State {\bf return} $\Gee$
        \EndFunction
    \end{algorithmic}
\end{algorithm}

\begin{algorithm}[H]
    \caption{{\sc Label Samples} \& \agg}\label{alg:agg}
    \begin{algorithmic}[1]   \small
    \Function{\sc LabelSamples}{$\dee, \tau, c=2$}
        \For{$c\tau$ random samples $t$ from $\dee$}
            \State $l \gets$ {\sc PointQuery}$(t)$
            \State $\el.add(\langle t,l\rangle)$;
            $\dee.remove(t)$
        \EndFor
        \State {\bf return} $\dee,\el$
    \EndFunction
    
    \Function{\agg}{$\el,N,\tau,\g, multi=false$}
    \State sort $\g$ based on $\el.${\sc count}$(\gee)$, $\gee\in\g$, ascending
    \State $sum \gets 0$; $\Gee\gets\{\}$
    \For{$\gee \in \g$}  
        \State $E_\gee \gets \frac{\el.\mbox{\sc count}(\gee)}{|\el|} \times N$
        \State {\bf if:} {$multi=$ true} {\bf then} $\Gee \gets \{\gee' \; | \; \gee'$.parent$=\gee$.parent$\}$
            \State {\bf if:}{$sum + E_\gee < \tau$} {\bf then} $\Gee.${\sc add}$(\gee)$;
                $sum \gets sum + E_\gee$
            \State {\bf else: } $\g_{agg}.${add}$(\Gee)$ ; $\Gee\gets \{\gee\}$; $sum \gets E_\gee$
                
    \EndFor
    \State {\bf return} $\g_{agg}$.{add}$(\Gee)$
    \EndFunction

    \end{algorithmic}
\end{algorithm}

\begin{algorithm}[H]
    \caption{{\sc Base-Coverage}($\dee,\tau,\gee$)} \label{alg:pointcvg}
    \begin{algorithmic}[1] \small
    \Require{Dataset $\dee$, coverage threshold $\tau$, and target group $\gee$}
    \Ensure{Coverage of group $\gee$}
        \State $cnt \gets 0$
        \For{$\forall t \in \dee$}
            \State ans $\gets ${\sc AskQuestion}$(t,\gee)$
            \If{$ans=${\bf true}}   $cnt \gets cnt + 1$ \EndIf
            \State {\bf if} $cnt=\tau$ {\bf then return true} {\tt\small // covered}
        \EndFor
        \State {\bf return false} {\tt\small //uncovered}
    \end{algorithmic}
\end{algorithm}
\balance
}

\end{document}